\documentclass[11pt,reqno]{amsart}

\usepackage{amsfonts,amssymb,amsmath,gensymb,color,epic,eepic,float,fullpage}
\usepackage{empheq}
\usepackage[mathscr]{euscript}
\usepackage[T1]{fontenc}
\usepackage{mathtools}
\usepackage{comment}
\usepackage{bbm}
\usepackage{caption}
\usepackage{hyperref}
\usepackage{tikz}
\usetikzlibrary{graphs,patterns,decorations.markings,arrows,matrix}
\usetikzlibrary{calc,decorations.pathmorphing,decorations.markings,
decorations.pathreplacing,patterns,shapes,arrows}
\tikzstyle{block} = [rectangle,draw,text width=10em,text centered,rounded corners,minimum height=4em]
\tikzstyle{line} = [draw, -latex']
\allowdisplaybreaks
\newtheorem{defn}{Definition}
\newtheorem{lem}{Lemma}
\newtheorem{thm}{Theorem}

\newtheorem{cor}{Corollary}
\newtheorem{remark}{Remark}
\newtheorem{prop}{Proposition}
\newtheorem{ex}{Example}

\makeatletter
\DeclareFontFamily{OMX}{MnSymbolE}{}
\DeclareSymbolFont{MnLargeSymbols}{OMX}{MnSymbolE}{m}{n}
\SetSymbolFont{MnLargeSymbols}{bold}{OMX}{MnSymbolE}{b}{n}
\DeclareFontShape{OMX}{MnSymbolE}{m}{n}{
    <-6>  MnSymbolE5
   <6-7>  MnSymbolE6
   <7-8>  MnSymbolE7
   <8-9>  MnSymbolE8
   <9-10> MnSymbolE9
  <10-12> MnSymbolE10
  <12->   MnSymbolE12
}{}
\DeclareFontShape{OMX}{MnSymbolE}{b}{n}{
    <-6>  MnSymbolE-Bold5
   <6-7>  MnSymbolE-Bold6
   <7-8>  MnSymbolE-Bold7
   <8-9>  MnSymbolE-Bold8
   <9-10> MnSymbolE-Bold9
  <10-12> MnSymbolE-Bold10
  <12->   MnSymbolE-Bold12
}{}

\let\llangle\@undefined
\let\rrangle\@undefined
\DeclareMathDelimiter{\llangle}{\mathopen}%
                     {MnLargeSymbols}{'164}{MnLargeSymbols}{'164}
\DeclareMathDelimiter{\rrangle}{\mathclose}%
                     {MnLargeSymbols}{'171}{MnLargeSymbols}{'171}
\makeatother
\def\leq{\leqslant}
\def\geq{\geqslant}

\def\adag{a^\dagger}
\def\phid{\phi^\dagger}

\def\dg{\color{green!65!black}}
\def\bl{\color{blue!75!black}}

\newcommand{\bra}[1]{\left\langle #1\right|}
\newcommand{\ket}[1]{\left|#1\right\rangle}

\newcommand{\Ket}[1]{\left\Vert#1\right\rrangle}

\DeclareMathOperator{\Span}{Span}
\DeclareMathOperator{\Tr}{Tr}

\begin{document}

\title{A new generalisation of Macdonald polynomials}

\author{Alexandr Garbali, Jan de Gier and Michael Wheeler}
\address{ARC Centre of Excellence for Mathematical and Statistical Frontiers (ACEMS), School of Mathematics and Statistics, University of Melbourne, Parkville, Victoria 3010, Australia}
\email{alexandr.garbali@unimelb.edu.au, jdgier@unimelb.edu.au, wheelerm@unimelb.edu.au}

\maketitle

%abstract
\begin{abstract}
We introduce a new family of symmetric multivariate polynomials, whose coefficients are meromorphic functions of two parameters $(q,t)$ and polynomial in a further two parameters $(u,v)$. We evaluate these polynomials explicitly as a matrix product. At $u=v=0$ they reduce to Macdonald polynomials, while at $q=0$, $u=v=s$ they recover a family of inhomogeneous symmetric functions originally introduced by Borodin.
\end{abstract}

%intro

\section{Introduction}

\subsection{Background}

The Macdonald polynomials \cite{Macd88,MacdBook}, denoted $P_{\lambda}(x_1,\dots,x_n;q,t)$, are a celebrated basis for the ring of symmetric functions in $n$ variables.  They simultaneously generalise many important classes of symmetric functions, including the Schur, Hall--Littlewood and Jack polynomials, which can all be recovered by appropriate specialisations of the parameters $(q,t)$. Macdonald polynomials have been deeply influential in a variety of disciplines of mathematics, from the representation theory of affine Hecke algebras \cite{Cher95a,Cher95b,Opdam}, to Hilbert schemes \cite{Haiman}, to integrable stochastic systems \cite{BorodinC}. Despite the generality of Macdonald polynomials, they are themselves special cases of some even more general classes. These include the interpolation and Koornwinder polynomials, which are both examples of \textit{inhomogeneous} symmetric functions that include Macdonald polynomials at their leading degree.

The purpose of this paper is to introduce a new generalisation of Macdonald polynomials, denoted $P_{\lambda}(x_1,\dots,x_n;q,t;u,v)$, which have polynomial dependence on two additional parameters $u$ and $v$. Like the examples listed above, these functions are inhomogeneous symmetric polynomials in $(x_1,\dots,x_n)$, but the Macdonald polynomials do not occur as their top degree -- rather, they are embedded in $P_{\lambda}(x_1,\dots,x_n;q,t;u,v)$ as the constant term in $(u,v)$. The key characteristic of these polynomials is that they not only generalise Macdonald polynomials, they are also generalisations of a family of inhomogeneous symmetric functions $F_{\lambda}(x_1,\dots,x_n;t;s)$ recently studied by Borodin and Petrov \cite{Borodin,BorodinP}. Reducing our family of polynomials to known cases can be summarised by the following commutative diagram:
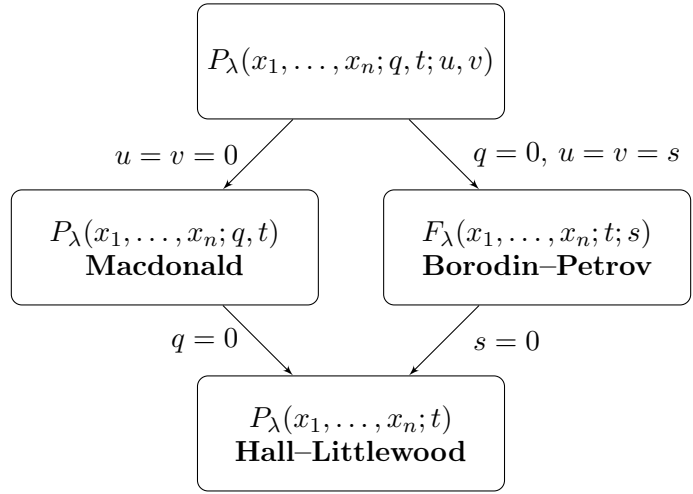
\begin{center}
\begin{tikzpicture}[node distance = 3.5cm,auto]
% nodes
\node [block] (P) {$P_{\lambda}(x_1,\dots,x_n;q,t;u,v)$};
\node [block,below right of=P] (B) 
{$F_{\lambda}(x_1,\dots,x_n;t;s)$\\ {\bf Borodin--Petrov}};
\node [block,below left of=P] (M) 
{$P_{\lambda}(x_1,\dots,x_n;q,t)$\\ {\bf Macdonald}};
\node[block,below right of=M] (HL)
{$P_{\lambda}(x_1,\dots,x_n;t)$\\ {\bf Hall--Littlewood}};
% arrows
\path [line] (P) -- node[right] {\ \ $q=0$, $u=v=s$} (B);
\path [line] (P) -- node[left] {$u=v=0$\ \ } (M);
\path[line] (M) -- node[left] {$q=0$\ \ } (HL); 
\path[line] (B) -- node[right] {\ \ $s=0$} (HL);
\end{tikzpicture}
\end{center}
We expect that the family of polynomials $P_{\lambda}(x_1,\dots,x_n;q,t;u,v)$ will prove to be important for several reasons: {\bf 1.} Both $P_{\lambda}(x_1,\dots,x_n;q,t)$ and, more recently, $F_{\lambda}(x_1,\dots,x_n;t;s)$ have been shown to be vital in integrable probability. Macdonald processes include a wide range of particle-hopping processes as their specialisations \cite{BorodinC}, and it appears that the stochastic vertex model used in the construction of $F_{\lambda}(x_1,\dots,x_n;t;s)$ plays a similarly powerful role, containing a number of sub-processes as special cases \cite{BorodinP}. The sheer existence of $P_{\lambda}(x_1,\dots,x_n;q,t;u,v)$ suggests that these two, somewhat complementary pictures could be unified. {\bf 2.} Given that both $P_{\lambda}(x_1,\dots,x_n;q,t)$ and $F_{\lambda}(x_1,\dots,x_n;t;s)$ enjoy a host of special properties, such as Cauchy identities, branching rules and Pieri identities, it is natural to expect that their mutual generalisation, $P_{\lambda}(x_1,\dots,x_n;q,t;u,v)$, will too. We plan to address such questions in a separate publication. {\bf 3.} The functions $F_{\lambda}(x_1,\dots,x_n;t;s)$ contain the (Grassmannian) Grothendieck polynomials as a special case \cite{Borodin-com}. This means that $P_{\lambda}(x_1,\dots,x_n;q,t;u,v)$ links Macdonald polynomials with polynomials that have $K$-theoretic content, going beyond the known reduction to cohomology (the reduction to Schur polynomials). 

Unlike the traditional approach in Macdonald theory, in which Macdonald polynomials are defined by a set of properties and then proven to exist, we shall instead write down an explicit formula for $P_{\lambda}(x_1,\dots,x_n;q,t;u,v)$. Our methodology is essentially the same as in \cite{CantinidGW}, where a matrix product formula was obtained for the Macdonald polynomials. A key ingredient of this approach is to construct a {\it non-symmetric} family of polynomials $f_{\mu}(x_1,\dots,x_n;q,t;u,v)$, where $\mu$ denotes a composition, which can be summed appropriately to produce the symmetric polynomial $P_{\lambda}(x_1,\dots,x_n;q,t;u,v)$. In this sense, the family $f_{\mu}(x_1,\dots,x_n;q,t;u,v)$ plays an analogous role to non-symmetric Macdonald polynomials \cite{Cher95b,Opdam}, although it should be emphasised that they are not the same as the latter, even at the special value $u=v=0$ of the parameters. Below we outline the basics of our construction.

\subsection{Layout of the paper} 

In Section \ref{se:f} we study polynomial representations of the type $A_{n-1}$ Hecke algebra and families of polynomials $f_{\mu}$ which satisfy local quantum Knizhnik--Zamolodchikov exchange relations. We show that, by summing $\mu$ over all permutations of a partition 
$\lambda$, one obtains a symmetric function $P_{\lambda}$. Section \ref{se:3} examines, in a general setting, how it is possible to construct such families $f_{\mu}$ as matrix products. As we show, the basic requirements for the construction are {\bf 1.} A suitable solution of the (higher-rank) Yang--Baxter algebra, and {\bf 2.} A suitable linear form which maps elements of the algebra to the space of polynomials in $n$ variables.

In Section \ref{se:new-L} we present a new solution of the Yang--Baxter algebra of generic rank $r$, in terms of the algebra of $t$-deformed bosons. We construct an $L$-matrix that satisfies the Yang--Baxter algebra, starting from a solution of Jimbo \cite{Jimbo86a} and applying an algebra homomorphism (the details are deferred to Section \ref{proofrll}). In Section \ref{se:poly}, we are then able to apply the general theory developed at the start of the paper to the specific solution of the Yang--Baxter algebra obtained in Section \ref{se:new-L}. This leads us to explicit matrix product formulae for both $f_{\mu}(x_1,\dots,x_n;q,t;u,v)$ and $P_{\lambda}(x_1,\dots,x_n;q,t;u,v)$, which are the main results of the paper. Section \ref{se:mac} proves the reduction to Macdonald polynomials at $u=v=0$, which is almost immediate by virtue of the results in \cite{CantinidGW}. Section \ref{se:bor-pet} proves the (much more challenging) reduction to the Borodin--Petrov polynomials at $q=0$, $u=v=s$.

\subsection{Notation and conventions}

A {\it composition} $\mu=(\mu_1,\dots,\mu_n)$ is an $n$-tuple of non-negative integers, and $\mu_i$ is its $i^{\rm th}$ {\it part}. A {\it partition} $\lambda=(\lambda_1,\dots,\lambda_n)$ is an $n$-tuple of non-negative integers, which satisfy $\lambda_1 \geq \cdots \geq \lambda_n \geq 0$. Given a composition $\mu$, we let $\mu^+$ denote the unique partition which can be obtained by reordering the parts of $\mu$. Throughout this paper $\lambda$ will always refer to a partition and $\mu$ to a composition. Furthermore, the largest part $\lambda_1$ of $\lambda$ will be denoted by $r$, for rank:
\begin{equation}
\label{eq:rank}
r=\lambda_1.
\end{equation}

\section{Families of non-symmetric polynomials and quantum Knizhnik--Zamolodchikov exchange relations}
\label{se:f}

Following \cite{CantinidGW,KasataniT}, we study families of non-symmetric polynomials which satisfy local quantum Knizhnik--Zamolodchikov exchange relations. The exchange relations are expressed via the action of generators of the Hecke algebra. Our aim here is to discuss such families at a completely general level, as we will only specialise to a particular non-symmetric family later on in the paper. 

\subsection{Polynomial representation of Hecke algebra}

We consider polynomial representations of the Hecke algebra of type $A_{n-1}$, with generators $T_{i}$ given by 
\begin{align}
\label{hecke-gen}
T_i
=
t-\frac{tx_i-x_{i+1}}{x_i-x_{i+1}}(1-\sigma_i),
\qquad
1 \leq i \leq n-1,
\end{align} 
where $\sigma_i$ is the transposition operator with action $(\sigma_i  p)(\dots,x_i,x_{i+1},\dots)= p(\dots,x_{i+1},x_i,\dots)$ on any polynomial $p(x_1,\dots,x_n)$. The operators \eqref{hecke-gen} provide a faithful representation of the Hecke algebra:
\begin{align}
\label{hecke-alg}
(T_{i}-t)(T_{i}+1)=0,
\qquad
T_{i} T_{i + 1} T_{i} = T_{i + 1} T_i T_{i + 1},
\qquad
T_i T_j = T_j T_i, \ |i-j| > 1.
\end{align}

\subsection{Non-symmetric polynomials and quantum Knizhnik--Zamolodchikov equations}

Consider for each partition $\lambda$ a family of polynomials indexed by compositions that are permutations of $\lambda$. We denote this family by $\{f_{\mu}(x_1,\dots,x_n)\}_{\mu^+=\lambda}$ and assume that it satisfies the following relations with respect to the generators \eqref{hecke-gen} of the Hecke algebra:
\begin{align}
\label{eq:exchange1}
&T_i f_{\mu_1,\dots,\mu_n}(x_1,\dots,x_n)
=
f_{\mu_1,\dots,\mu_{i+1},\mu_i,\dots,\mu_n}
(x_1,\dots,x_n),
\qquad\ 
\text{when}\ \ 
\mu_i > \mu_{i+1},
\\
\label{eq:exchange2}
&T_i f_{\mu_1,\dots,\mu_n}(x_1,\dots,x_n)
=
t f_{\mu_1,\dots,\mu_{i+1},\mu_i,\dots,\mu_n}
(x_1,\dots,x_n),
\qquad
\text{when}\ \ 
\mu_i = \mu_{i+1}.
\end{align}
These two relations play quite different roles. Equation \eqref{eq:exchange1} expresses one member of the family in terms of another, where the two members are related by a simple transposition of their indexing composition. On the other hand, \eqref{eq:exchange2} dictates that $f_{\mu}$ is symmetric in $(x_i,x_{i+1})$ if $\mu_i = \mu_{i+1}$.

The relations \eqref{eq:exchange1}--\eqref{eq:exchange2} {\it do not} uniquely determine the family $\{f_{\mu}(x_1,\dots,x_n)\}_{\mu^+=\lambda}$. Indeed, by choosing $f_{\lambda}$ to be any polynomial which is symmetric in $(x_i,x_{i+1})$ if $\lambda_i = \lambda_{i+1}$, then acting with \eqref{eq:exchange1} to build up the entire family, one finds that \eqref{eq:exchange1}--\eqref{eq:exchange2} hold generally. By supplementing these relations by appropriate boundary conditions (such as, for example, a cyclic property \cite{CantinidGW,KasataniT}) the family can be made unique, but this will not concern us in the present work.

\subsection{Symmetric polynomials}
\label{se:symmpol}

Motivated by the theory of non-symmetric Macdonald polynomials, we now consider polynomials $P_{\lambda}(x_1,\dots,x_n)$ which are obtained by summing over all $f_{\mu}$, such that $\mu$ lies in the Weyl orbit of the partition $\lambda$:
\begin{align*}
P_{\lambda}(x_1,\dots,x_n)
:=
\sum_{\mu: \mu^{+}=\lambda}
f_{\mu}(x_1,\dots,x_n).
\end{align*}
By virtue of the exchange relations \eqref{eq:exchange1}--\eqref{eq:exchange2}, we can easily deduce the following property of $P_\lambda$.

\begin{lem}{\rm
The polynomial $P_{\lambda}(x_1,\dots,x_n)$ is symmetric in $(x_1,\dots,x_n)$.
}
\end{lem}

\begin{proof}
We need to show that $T_i P_\lambda = t P_\lambda$ for all $i=1,\ldots,n-1$, since this would imply symmetry in $(x_1,\dots,x_n)$. Acting with $T_i$ on \eqref{eq:exchange1} we find that, for $\mu_i < \mu_{i+1}$,
\begin{align*}
T_i 
f_{\ldots,\mu_i,\mu_{i+1},\ldots} 
&= 
T_i^2  
f_{\ldots,\mu_{i+1},\mu_{i},\ldots} 
= 
\left(t+(t-1)T_i \right)  
f_{\ldots,\mu_{i+1},\mu_{i},\ldots} 
= 
t f_{\ldots,\mu_{i+1},\mu_{i},\ldots} 
+
(t-1) f_{\ldots,\mu_{i},\mu_{i+1},\ldots}.
\end{align*}
We thus find
\begin{align*}
T_i \sum_\mu f_\mu 
&= 
\sum_{\mu: \mu_i < \mu_{i+1}} 
\left( t f_{\sigma_i\mu} +(t-1)f_\mu \right) 
+ 
\sum_{\mu: \mu_i = \mu_{i+1}} 
t f_\mu 
+ 
\sum_{\mu: \mu_i > \mu_{i+1}} 
f_{\sigma_i\mu} 
\nonumber
\\
&= 
\sum_{\mu: \mu_i < \mu_{i+1}} 
t f_{\sigma_i\mu} 
+ 
\sum_{\mu:\mu_i \leq \mu_{i+1}} 
t f_\mu = t \sum_\mu f_\mu.
\end{align*}
\end{proof}
We have thus shown that if we have a family $\{f_{\mu}(x_1,\dots,x_n)\}_{\mu^+=\lambda}$ of non-symmetric polynomials obeying \eqref{eq:exchange1}--\eqref{eq:exchange2}, then the polynomial $P_{\lambda}$, obtained by summing over all members of the family, is symmetric. This result is the foundation which allows us to construct the new family 
$P_{\lambda}(x_1,\dots,x_n;q,t;u,v)$.

\section{Matrix product expression}
\label{se:3}

In this section we explain a general construction to obtain explicit families of polynomials that satisfy the relations \eqref{eq:exchange1}--\eqref{eq:exchange2}, using solutions of the Yang--Baxter algebra.

\subsection{Matrix product expression for $f_{\mu}$ and Zamolodchikov--Faddeev algebra}
\label{se:mpa}
We begin by writing explicitly the higher-rank $R$-matrices which, in Jimbo's classification \cite{Jimbo86b}, are solutions of the $U_{t^{1/2}}(A_r^{(1)} )$ Yang--Baxter equation\footnote{We refrain from using the parameter $q$ when writing the quantum group, since this would create confusion with the $q$ parameter in Macdonald polynomials.}:   
\begin{multline}
\label{Rmat-def}
R(z)
=
\sum_{i=0}^{r}
E_{i,i} \otimes E_{i,i} 
\\
+
\sum_{0 \leq i < j \leq r}
\left(
b_{+}(z)
E_{i,i} \otimes E_{j,j}
+
b_{-}(z)
E_{j,j} \otimes E_{i,i}
+
c_{+}(z)
E_{i,j} \otimes E_{j,i}
+
c_{-}(z)
E_{j,i} \otimes E_{i,j}
\right),
\end{multline}
where $E_{i,j}$ is the matrix with a $1$ at position $(i,j)$ and zeros everywhere else, and the matrix entries are given by
\begin{align*}
b_{+}(z)
=
\frac{1-z}{1-tz},
\qquad
b_{-}(z)
=
\frac{t(1-z)}{1-tz},
\qquad
c_{+}(z)
=
\frac{1-t}{1-tz},
\qquad
c_{-}(z)
=
\frac{(1-t)z}{1-tz}.
\end{align*}
The $R$-matrix \eqref{Rmat-def} is in fact twisted, in the sense of {\it Drinfeld twists,} in such a way that all its columns sum to $1$. It therefore generalises the stochastic six-vertex model to arbitrary rank. We define from this the $\check{R}$-matrix, given by $\check{R}(z) =P R(z)$, where $P$ is the $(r+1)^2 \times (r+1)^2$ permutation matrix.

Now assume that there exist linear operators $A_i(x)\ (i=0,1,\ldots,r)$ acting on some vector space $\mathcal{F}$, a linear form $\rho: {\rm End}(\mathcal{F}) \rightarrow \mathbb{C}[x_1,\ldots,x_n]$, and define for all compositions $\mu$ with largest part $r$ the polynomial
\begin{align}
f_\mu(x_1,\ldots,x_n) := \rho\left( A_{\mu_1}(x_1) \cdots A_{\mu_n}(x_n) \right).
\label{eq:mpa}
\end{align}
It is easy to show \cite{CantinidGW,CrampeRV} that such a family $\{f_{\mu}\}_{\mu^{+} = \lambda}$ satisfies \eqref{eq:exchange1}--\eqref{eq:exchange2}, provided that the operators $A_i(x)$ obey the Zamolodchikov--Faddeev (ZF) algebra \cite{Fad1980,ZZ1979}:
\begin{align}
\check{R}(x/y)\cdot\left [\mathbb{A}(x)\otimes \mathbb{A}(y)\right] 
= 
\left [\mathbb{A}(y)\otimes \mathbb{A}(x)\right] ,
\label{eq:ZFdef}
\end{align}
where $\check{R}(x/y)$ is the $\check{R}$-matrix based on $U_{t^{1/2}}(A_r^{(1)} )$, and 
$\mathbb{A}(x)$ is an $(r+1)$-dimensional operator valued column vector given by
\[
\mathbb{A}(x) = (A_0(x),\ldots, A_r(x))^{\rm T}.
\]

\subsection{Solutions of the ZF algebra from the Yang--Baxter algebra}

One way to obtain a set of operators $A_i(x)$ that satisfy the ZF relations \eqref{eq:ZFdef} is to inherit them from a solution to the Yang--Baxter algebra. The Yang--Baxter algebra is the set of bilinear relations which are encoded by the equation
\begin{align}
\label{rll-eq}
\check{R}(x/y) \cdot \left[L(x)\otimes L(y)\right]
=
\left[L(y)\otimes L(x)\right] \cdot \check{R}(x/y),
\end{align}
where $L(x)$ is an $(r+1) \times (r+1)$ matrix, whose entries are operators acting on $\mathcal{F}$. More generally, the Yang--Baxter algebra also applies to any product of $L$-matrices which satisfy \eqref{rll-eq}. For example, if we construct the rank-$r$ monodromy matrix
\begin{align}
\label{monodromy}
T(x)
=
L^{(1)}(x)
\cdots
L^{(r)}(x),
\end{align}
where each $L$-matrix $L^{(i)}(x)$ satisfies \eqref{rll-eq} and acts on a separate copy $\mathcal{F}^{(i)}$ of the Hilbert space, then it immediately follows that
\begin{align}
\check{R}(x/y) \cdot \left[T(x)\otimes T(y)\right]
=
\left[T(y)\otimes T(x)\right] \cdot \check{R}(x/y).
\end{align}
Solutions of the ZF algebra may then be obtained as follows:
\begin{prop}{\rm
If $\mathbb{A}(x)$ is identified with any column of $T(x)$ as defined in \eqref{monodromy}, then \eqref{eq:ZFdef} holds.
}
\end{prop}

\section{A new $L$-matrix}
\label{se:new-L}

In this section we present a new solution to the Yang--Baxter algebra \eqref{rll-eq}. To formulate this solution we first introduce $t$-deformed bosonic operators.

\subsection{The algebra of $t$-bosons}

The $t$-boson algebra $\mathcal{B} = \langle \phi,\phid,k\rangle$ is generated by three operators $\{\phi,\phid,k\}$, that satisfy the bilinear relations
\begin{align}
\label{t-bos}
\phi k = t k \phi,
\qquad
t \phid k = k \phid,
\qquad
\phi \phid - t \phid \phi = 1-t.
\end{align}
Define vector spaces $\mathcal{F}=\Span\{\ket{m}\}_{m=0}^\infty$ and $\mathcal{F}^*=\Span\{\bra{m}\}_{m=0}^\infty$, which will be the representation spaces for the $t$-boson algebra. We use the Fock and dual Fock representation of the algebra \eqref{t-bos}:
\begin{align*}
\phi \ket{m} = (1-t^m) \ket{m-1},
\qquad
\phid \ket{m} = \ket{m+1},
\qquad
k \ket{m} = t^m \ket{m},
\\
\nonumber
\bra{m} \phi = (1-t^{m+1}) \bra{m+1},
\qquad
\bra{m} \phid = \bra{m-1},
\qquad
\bra{m} k = t^m \bra{m}.
\end{align*}

\subsection{Some important remarks on notation}

It will be necessary to use $r^2$ commuting copies of the $t$-boson algebra \eqref{t-bos}. We shall distinguish these copies by the use of subscripts and superscripts and write them as $\mathcal{B}_i^{(j)} = \langle \phi_i^{(j)},\phi^{\dagger(j)}_i,k_i^{(j)} \rangle$, where $1 \leq i,j \leq r$. The operators in two algebras $\mathcal{B}_a^{(b)}$ and $\mathcal{B}_c^{(d)}$ mutually commute, unless $a=c$ and $b=d$, in which case the two algebras are identically equivalent. 

At all times, we use subscripts $i$ to distinguish between different {\it families} of bosons. There will be $r$ different families, and accordingly $1 \leq i \leq r$. Superscripts $j$, on the other hand, are used to indicate bosons with occur in the $j^{\rm th}$ $L$-matrix in the product \eqref{monodromy}. When it is not important to specify from which $L$-matrix the bosons come, we will omit the superscript to lighten the notation.

\subsection{A higher rank solution of the intertwining equation}

One of the key results in this paper is a new solution of the Yang--Baxter algebra \eqref{rll-eq} in terms of the algebra \eqref{t-bos} of $t$-bosons. As we discuss below, it generalises some known solutions of \eqref{rll-eq} to arbitrary values of the parameters $u$, $v$ and of the rank $r$. We define an $L$-matrix as follows:
\begin{align}
\label{eq:Lmat}
\nonumber
L_{00}
=
1-u x\prod_{l=1}^{r} k_l,
\qquad
&
L_{0j}
=
\left(1-uv\prod_{l=1}^{r} k_l\right ) \phi_j,\ \ 
\text{for }\ 1 \leq j \leq r,
\qquad
L_{i0}(x)
=
x 
\left( \prod_{l=i+1}^{r} k_l \right)
\phid_i,
\\
\nonumber
\\
&
L_{ij}(x)
=
\left\{
\begin{array}{ll}
(x-v k_i) \prod_{l=i+1}^{r}
k_l,
\qquad
&
i=j
\\
\\
x  
\left(\prod_{l=i+1}^{r} k_l \right)
\phid_i \phi_j ,
\qquad
&
i>j
\\
\\
v \left( \prod_{l=i+1}^{r}
k_l\right)
\phid_i  \phi_j,
\qquad
&
i<j
\end{array}
\right.
\quad\ \text{for }\ 
1 \leq i \leq r,\ 
1 \leq j \leq r,
\end{align}
where, as mentioned above, subscripts designate $r$ different families of bosons.

\begin{remark}{\rm
For $r= 1$ and $u=-1$ this $L$-matrix was introduced in the context of integrable stochastic models by Povolotsky \cite{Povo13}, see also Corwin and Petrov \cite{CorwinP} for a generalisation to higher spin versions. It also essentially appears in the explicit formulas of the $R$-matrix and the $Q$-operators in the work of Mangazeev \cite{Mang14a,Mang14b}. 

For $u = v = 0$ the model reduces to the $r$-species $t$-boson process appearing in the works of Prolhac et al. \cite{ProlhacEM} and Arita et al. \cite{AritaAMP} in the context of matrix product states for the asymmetric exclusion process, and in the works of Inoue et al. \cite{InoueKO}, Takeyama \cite{Takey14,Takey15}, inspired by the $r = 1$ case in \cite{SasaW98}, and Tsuboi \cite{Tsuboi}. In particular it reduces to the $L$-matrix used in \cite{CantinidGW} to construct a matrix product expression for Macdonald polynomials.

A characterisation of a rank-$r$ and higher spin $R$-matrix is presented in the recent work of Kuniba et al. \cite{KunibaMMO}, but no explicit form was given. We expect that our $L$-matrix is related to this work.

}
\end{remark}

\begin{thm}
\label{rll}{\rm
The $L$-matrix defined in \eqref{eq:Lmat} satisfies the Yang-Baxter algebra \eqref{rll-eq}.
}
\end{thm}

\begin{proof}
The proof of Theorem~\ref{rll} will be deferred to Section~\ref{proofrll} in order to not interrupt the flow of the paper.
\end{proof}

\begin{ex}
\label{ex-rank1}
{\rm
In the case $r=1$, the $\check{R}$ and $L$-matrices are given by
\begin{align*}
\check{R}(z)
&=
\left(
\begin{array}{cc|cc}
1 & 0 & 0 & 0
\\
0 & c_{-} & b_{-} & 0
\\
\hline
0 & b_{+} & c_{+} & 0
\\
0 & 0 & 0 & 1
\end{array}
\right),
\quad
L(x)
=
\begin{pmatrix}
1-x u k & (1-uv k) \phi
\\
x \phid & x-v k
\end{pmatrix},
\end{align*}
where we have omitted the bosonic subscripts, given that only one family is present. In this case the $R$-matrix is that of the stochastic six-vertex model, and the $L$-matrix has appeared in various forms in the literature \cite{Povo13,Mang14a,Mang14b,Borodin,CorwinP,BorodinP}. Here we adopt an operatorial version of the entries and include two deformation parameters $u,v$ (for example, one sets $u=v=s$ to recover the $L$-matrix of \cite{Borodin,BorodinP}).
}
\end{ex}

\begin{ex}
\label{ex:rank2}
{\rm
In the case $r=2$, the $\check{R}$ and $L$-matrices are given by
\begin{align*}
\check{R}(z)
=
\left(
\begin{array}{ccc|ccc|ccc}
1 & 0 & 0 & 0 & 0 & 0 & 0 & 0 & 0
\\
0 & c_{-} & 0 & b_{-} & 0 & 0 & 0 & 0 & 0
\\
0 & 0 & c_{-} & 0 & 0 & 0 & b_{-} & 0 & 0
\\
\hline
0 & b_{+} & 0 & c_{+} & 0 & 0 & 0 & 0 & 0
\\
0 & 0 & 0 & 0 & 1 & 0 & 0 & 0 & 0
\\
0 & 0 & 0 & 0 & 0 & c_{-} & 0 & b_{-} & 0
\\
\hline
0 & 0 & b_{+} & 0 & 0 & 0 & c_{+} & 0 & 0
\\
0 & 0 & 0 & 0 & 0 & b_{+} & 0 & c_{+} & 0
\\
0 & 0 & 0 & 0 & 0 & 0 & 0 & 0 & 1
\end{array}
\right),
\end{align*}
\begin{align*}
L(x)
=
\begin{pmatrix}
1-x u k_1 k_2 & (1-uv k_1 k_2) \phi_1 & (1-uv k_1 k_2) \phi_2
\\
x k_2 \phid_1 & (x-v k_1)k_2 & v k_2 \phid_1 \phi_2
\\
x \phid_2 & x \phid_2 \phi_1 & x - v k_2
\end{pmatrix}.
\end{align*}
}
\end{ex}

\section{A new class of symmetric polynomials}
\label{se:poly}

In Section~\ref{se:mpa} we gave a general construction of families of polynomials $\{f_{\mu}\}_{\mu^{+} = \lambda}$ which obey \eqref{eq:exchange1}--\eqref{eq:exchange2}, given a solution $L(x)$ of the Yang--Baxter algebra. We now apply this formalism to a specific example -- namely, to the $L$-matrices studied in Section \ref{se:new-L}, in order to obtain our new classes of polynomials $f_{\mu}$ and $P_{\lambda}$ as explicit matrix products.

\subsection{Matrix product expression for $f_{\mu}$}
\label{sse:f}

Define, as in equation \eqref{monodromy}, a monodromy matrix $T(x)$ whose constituent $L$-matrices are given by \eqref{eq:Lmat}. From this, construct a solution to the ZF algebra \eqref{eq:ZFdef}, by extracting the first column of $T(x)$:
\begin{align}
\label{ZF-ops}
\mathbb{A}(x) = (A_0(x),A_1(x), \ldots, A_r(x))^{\rm T} = (T_{00}(x),T_{10}(x),\ldots,T_{r0}(x))^{\rm T}.
\end{align}
For any composition $\mu$ with largest part $r=\mu_1^+$ we then define
\begin{align}
\label{mat-prod}
f_\mu(x_1,\ldots,x_n;q,t;u,v) := \rho \left( A_{\mu_1}(x_1) \cdots A_{\mu_n}(x_n) \mathbb{S} \right),
\end{align}
where the linear form $\rho$ will be given below. Here we work in slightly greater generality than in Section~\ref{se:mpa}, and aside from parameters $t,u,v$ which enter via $L(x)$, we also allow for an additional parameter $q$ which is incorporated via a twist matrix $\mathbb{S}$: 
\begin{align}
\label{Sdef} 
\mathbb{S}=S^{(1)} \cdots S^{(r)},
\qquad 
S^{(i)}= \left( \prod_{j=i+1}^{r} k_{j}^{(j-i)\alpha}\right)^{(i)}
\quad
\text{where}\ \ t^{\alpha} \equiv q.
\end{align} 
By their very construction, the $f_\mu$ obey the quantum Knizhnik--Zamolodchikov equations \eqref{eq:exchange1}--\eqref{eq:exchange2}.

It remains to specify the linear form. We let $\ket{m}_i^{(j)}$ denote a basis state in the Fock space corresponding to $\mathcal{B}_i^{(j)} = \langle \phi_i^{(j)},\phi^{\dagger(j)}_i,k_i^{(j)} \rangle$. We use abbreviated notation for a sum over all basis states: $\ket{\theta}_i^{(j)} = \sum_{m=0}^{\infty} \ket{m}_i^{(j)}$. The linear form $\rho$ is defined in the following way:
\begin{enumerate}
\item Trace over the Fock representation of all algebras $\mathcal{B}_i^{(j)}$ such that $i>j$.

\item Sandwich between vacuum states $\bra{0}_i^{(j)}$ and $\ket{0}_i^{(j)}$ for all algebras $\mathcal{B}_i^{(j)}$ such that $i<j$. 

\item Sandwich between the states $\bra{\theta}_i^{(i)}$ and $\ket{0}_i^{(i)}$ for all algebras 
$\mathcal{B}_i^{(i)}$.
\end{enumerate}
More succinctly, the form can be written as
\begin{multline}
\label{linearform}
\rho \left( A_{\mu_1}(x_1) \cdots A_{\mu_n}(x_n) \mathbb{S}\right) 
:= 
\\
\prod_{i=1}^{r}
\bra{0_1,\dots,0_{i-1},\theta_i}^{(i)}
\Tr \Big[ A_{\mu_1}(x_1) \cdots A_{\mu_n}(x_n) \mathbb{S} \Big]^{(i)}_{(i,\dots,r]}
\ket{0_1,\dots,0_i}^{(i)}.
\end{multline}
Although the $\mathcal{B}_i^{(i)}$ algebras seemingly bring considerable complication to \eqref{linearform}, it is not hard to show that only a single term in the sum $\sum_{m=0}^{\infty} \bra{m}_i^{(i)}$ survives -- namely, the term $m = m_i(\mu)$, where $m_i(\mu)$ is the part-multiplicity function:
\begin{align*}
m_i(\mu)
=
\#\{ \mu_k : \mu_k = i \}.
\end{align*}
This allows us to write, equivalently,
\begin{multline}
\label{linearform2}
\rho \left( A_{\mu_1}(x_1) \cdots A_{\mu_n}(x_n) \mathbb{S}\right) 
= 
\\
\prod_{i=1}^{r}
\bra{0_1,\dots,0_{i-1},m_i(\mu)}^{(i)}
\Tr \Big[ A_{\mu_1}(x_1) \cdots A_{\mu_n}(x_n) \mathbb{S} \Big]^{(i)}_{(i,\dots,r]}
\ket{0_1,\dots,0_i}^{(i)}.
\end{multline}
Although \eqref{linearform2} is manifestly simpler, \eqref{linearform} is preferable as the definition of the linear form $\rho$, since it does not depend on $\mu$.

At this stage, it is by no means obvious that this particular linear form is the best choice available. The main reason that we adopt it is that it succeeds -- it ultimately leads to a family of polynomials which simultaneously generalise both those of Macdonald and those of Borodin--Petrov -- a fact that will be borne out below. For now, let us only remark that the most natural choice for $\rho$, which would be to trace all $r^2$ bosons that appear, causes $f_{\mu}$ to vanish for all non-zero compositions $\mu$.  

\subsection{Matrix product expression for $P_{\lambda}$}
\label{se:P}

Following the procedure in Section \ref{se:symmpol}, we now obtain symmetric polynomials $P_{\lambda}$ by summing over all $\mu$ which are permutations of $\mu^{+} \equiv \lambda$. Summing over \eqref{linearform2}, we obtain
\begin{multline}
\label{symmP}
P_{\lambda}(x_1,\dots,x_n;q,t;u,v)
=
\\
\Omega_{\lambda}(q,t)
\times
\prod_{i=1}^{r}
\bra{0_1,\dots,0_{i-1},m_i(\lambda)}^{(i)}
\Tr \Big[ A(x_1) \cdots A(x_n) \mathbb{S} \Big]^{(i)}_{(i,\dots,r]}
\ket{0_1,\dots,0_i}^{(i)}
\end{multline}
where we have defined
\begin{align*}
A(x)=\sum_{j=0}^r A_j(x),
\end{align*}
and $\Omega_{\lambda}(q,t)$ is an introduced overall normalisation to be given in the next subsection. It is apparent that the product $A(x_1) \cdots A(x_n)$ gives rise to terms $A_{\mu_1}(x_1) \cdots A_{\mu_n}(x_n)$ for {\it all} compositions $\mu$ contained in the $n \times r$ rectangle, and not just those for which $\mu^{+} = \lambda$. However, any $\mu$ for which $\mu^{+} \not= \lambda$ gives a vanishing contribution to \eqref{symmP}. Note that $P_\lambda$ written in the form \eqref{symmP} is manifestly symmetric because $[A(x),A(y)]=0$, which follows from left-multiplying the ZF equation \eqref{eq:ZFdef} with the row vector $(1,1,\ldots,1)$ and the fact that the columns of the $\check{R}$-matrix add up to $1$.

Equation \eqref{symmP} is the main result of the paper. It is a completely explicit formula for the new family of symmetric polynomials $P_{\lambda}(x_1,\dots,x_n;q,t;u,v)$, whose specialisations will be explored in the coming sections.

\subsection{Specifying the normalisation}

Equation \eqref{symmP} contains a normalising factor that we can freely choose, without spoiling the symmetry in $(x_1,\dots,x_n)$. Bearing in mind that we will subsequently specialise \eqref{symmP} to Macdonald polynomials, we define 
\begin{align*}
\Omega_{\lambda}(q,t)
=
\prod_{1 \leq i<j \leq r}
\left(
1-q^{j-i} t^{\lambda'_i-\lambda'_j} 
\right),
\end{align*}
where $\lambda'$ denotes the {\it conjugate partition} of $\lambda$:
\begin{align*}
\lambda'_i - \lambda'_{i+1}= m_{i}(\lambda),
\quad\quad
\forall\ i \geq 1.
\end{align*}
This is the same normalisation as was used in \cite{CantinidGW}.

\subsection{A polynomial example}
\label{se:polexample}
We look at an explicit example for rank 2. Using \eqref{monodromy} we construct a solution of the ZF algebra by taking the first column of the monodromy matrix defined by
\begin{align*}
T(x) = L^{(1)}(x)\cdot L^{(2)}(x),
\end{align*}
where each $L$-matrix is a copy of the $3\times 3$ rank 2 matrix of Example~\ref{ex:rank2}. We thus have
\begin{align}
\label{exampleA}
\mathbb{A}(x)=\begin{pmatrix}
A_0(x)\\
A_1(x)\\
A_2(x)
\end{pmatrix}
=
\begin{pmatrix}
1-x u k_1 k_2 & (1-uv k_1 k_2) \phi_1 & (1-uv k_1 k_2) \phi_2
\\
x k_2 \phid_1 & (x-v k_1)k_2 & v k_2 \phid_1 \phi_2
\\
x \phid_2 & x \phid_2 \phi_1 & x - v k_2
\end{pmatrix}^{(1)}\cdot
\begin{pmatrix}
1-x u k_1 k_2
\\
x k_2 \phid_1
\\
x \phid_2
\end{pmatrix}^{(2)},
\end{align}
where we only wrote the first column of $L^{(2)}(x)$. We now explain the matrix product form for $\mu=(2,1)$, with $m_1(\mu)=m_2(\mu)=1$.

There are four boson families $\mathcal{B}_i^{(j)}$ with $i,j \in \{1,2\}$. According to the prescription \eqref{linearform2}, the two diagonal families $\mathcal{B}_i^{(i)}$ should be sandwiched between $\bra{m_i(\mu)}$ and $\ket{0}$, which for both $i=1$ and $i=2$ results in sandwiching between $\bra{1}$ and $\ket{0}$. The family $\mathcal{B}_1^{(2)}$ is sandwiched between $\bra{0}$ and $\ket{0}$ and the fourth family $\mathcal{B}_2^{(1)}$ will be traced over. Hence we define
\begin{multline}
f_{21}(x_1,x_2;q,t;u,v) :=\rho \big( A_2(x_1)A_1(x_2) \mathbb{S}\big) \\
= \bra{1_1}^{(1)} \bra{0_1,1_2}^{(2)}\ \Tr \Big[ A_2(x_1)A_1(x_2) \mathbb{S} \Big]_2^{(1)}\ \ket{0_1}^{(1)} \ket{0_1,0_2}^{(2)}.
\end{multline}
Note now that \eqref{exampleA} only contains the creation operator $\phi_1^{\dag(2)}$ of $\mathcal{B}_1^{(2)}$ and not the annihilation operator $\phi_1^{(2)}$. As $\mathcal{B}_1^{(2)}$ is sandwiched between $\bra{0}$ and $\ket{0}$, the nonzero remaining terms are those not containing $\phi_1^{\dag(2)}$. In other words, we can set $\phi_1^{\dag(2)}=0$ and $k_1^{(2)}=1$ in \eqref{exampleA}.

To ease notation we call the diagonal families $\mathcal{B}_i^{(i)}=\langle a_i,a_i^{\dag},\kappa_i \rangle$ and drop the indices from the remaining family $\mathcal{B}_2^{(1)}$, and find after projecting out the family $\mathcal{B}_1^{(2)}$ that we have
\begin{align}
\label{examplef2}
f_{21}(x_1,x_2;q,t;u,v)
= \bra{1_1} \bra{1_2} \ \Tr_\phi \Big[ \tilde{A}_2(x_1)\tilde{A}_1(x_2) \mathbb{S} \Big] \ \ket{0_1} \ket{0_2},
\end{align}
where $\tilde{A}_i(x)$ are determined from
\begin{multline}
\begin{pmatrix}
\tilde{A}_0(x)\\
\tilde{A}_1(x)\\
\tilde{A}_2(x)
\end{pmatrix}=
\begin{pmatrix} 1-x u \kappa_1 k & (1-uv \kappa_1 k) \phi \\ x k a_1^\dag & vk a_1^\dag \phi \\ x \phid & x-vk \end{pmatrix} 
\cdot 
\begin{pmatrix} 1-xu \kappa_2 \\ x a_2^{\dag} \end{pmatrix} 
\\
=
\begin{pmatrix} (1-x u \kappa_1 k)(1-xu \kappa_2) + x (1-uv \kappa_1 k) \phi a_2^{\dag}\\
x k a_1^{\dag} \big( 1-xu \kappa_2 + v \phi a_2^{\dag} \big) \\
x \big( \phid (1-xu \kappa_2) + (x-vk) a_2^{\dag}\big) \end{pmatrix}.
\label{eq:ZFr=3}
\end{multline}
The projection of the $a$ bosons in \eqref{examplef2} implies that we only need to collect terms in the product $\tilde{A}_2(x_1)\tilde{A}_1(x_2)$ that are proportional to $a_1^{\dag}a_2^{\dag}$, as other terms project to zero in the bra-ket between $\bra{1_1}\bra{1_2}$ and $\ket{0_1}\ket{0_2} $. The surviving terms are
\begin{align*}
x_1x_2 \Tr_\phi \big[ \big( v \phid (1-x_1u \kappa_2) k a_1^{\dag} \phi a_2^{\dag} + (x_1-vk) a_2^{\dag} k a_1^{\dag} (1-x_2 u \kappa_2 ) \big) k^\alpha \big] ,
\end{align*}
where we also used the definition \eqref{Sdef} of the twist $\mathbb{S}$.

The next step is to order both $\adag$ bosons to the left and pair up $\phid$ and $\phi$. This will result in some additional factors of $t$ due to commutation relations between $a^\dag$ and $\kappa$, and between $\phi$ and $k$. For this example, after projecting out the $\adag$ bosons we arrive at
\begin{align*}
f_{21}(x_1,x_2;q,t; u,v) &= x_1x_2 \Tr_\phi \big[ \big( v t^{-1} \phid \phi k (1-x_1u t) + (x_1-vk) k (1-x_2 u ) \big) k^{\alpha} \big] \nonumber\\
&= x_1x_2 \Tr_\phi \big[ \big( v t^{-1} (1-k) k (1-x_1u t ) + (x_1-vk) k (1-x_2 u ) \big) k^{\alpha} \big]\nonumber \\
&= x_1x_2 \Tr_\phi \big[ v t^{-1} (1-k-tk)k^{1+\alpha} + x_1( 1-uv(1-k) )k^{1+\alpha} + x_2 uv k^{2+\alpha} - x_1x_2 uk^{1+\alpha} \big].
\end{align*}
The traces can now be simply evaluated because $\Tr_\phi k^\beta = (1-t^\beta)^{-1}$, resulting in
\begin{multline}
f_{21}(x_1,x_2;q,t; u,v)= 
\\
x_1x_2\left( -\frac{v(1-q)}{(1-qt)(1-qt^2)} +\frac{1-qt^2 -uvqt(1-t)}{(1-qt)(1-qt^2)} x_1 
+  \frac{uv}{1-qt^2}x_2 - \frac{u}{1-qt}x_1x_2\right).
\end{multline}

Likewise we compute
\begin{align*}
f_{12}(x_1,x_2;q,t;u,v) &=\rho \big( A_1(x_1)A_2(x_2) \mathbb{S}\big),
\end{align*}
where again we need to consider only the terms proportional to $\phi_1^{\dag(1)}\phi_2^{\dag(2)} \equiv\adag_1\adag_2$, which in this case are given by
\begin{align*}
x_1x_2 \Tr_\phi \big[ \big( v k \adag_1 \phi \adag_2 \phid (1-x_2u \kappa_2) + k (1-x_1 u \kappa_2) (x_2-vk) \adag_1 \adag_2 \big) k^\alpha \big] .
\end{align*}
Reordering and projecting out the $\adag$ bosons we get in this case
\begin{align*}
f_{12}(x_1,x_2;q,t;u,v) &= x_1x_2 \Tr_\phi \big[ \big( v k \phi \phid (1-x_2u) + k (1-x_1 ut) (x_2-vk) \big) k^\alpha \big] . \nonumber\\
&=x_1x_2 \Tr_\phi \big[ \big( v k (1-tk) (1-x_2u) + k (1-x_1 ut) (x_2-vk) \big) k^\alpha \big] \nonumber \\
&= x_1x_2 \Tr_\phi \big[ v (1-tk-k)k^{1+\alpha} + x_1 uvt k^{2+\alpha} + x_2 (1-uv(1-tk)) k^{1+\alpha} - x_1x_2 ut k^{1+\alpha} \big].
\end{align*}
After taking the traces we end up with
\begin{multline}
f_{12}(x_1,x_2;q,t;u,v)=
\\
x_1x_2\left( -\frac{v(1-q)t}{(1-qt)(1-qt^2)} + \frac{uv t}{1-qt^2}x_1
+ \frac{1-qt^2-uv(1-t)}{(1-qt)(1-qt^2)} x_2  - \frac{ut}{1-qt}x_1x_2\right).
\end{multline}
Finally, the symmetric polynomial $P_{21}=\Omega_{21}(f_{21}+f_{12})$, where
$ \Omega_{21}=1-qt$, is equal to
\begin{align}
P_{21}(x_1,x_2;q,t;u,v) 
= 
\left(1- \frac{uv(1-q)t}{1-qt^2}  \right)
(x_1^2x_2+x_1x_2^2)
-
(1+t) 
\left( \frac{v(1-q)}{1-qt^2} x_1x_2 +u x_1^2x_2^2\right).
\end{align}

\section{Specialisation to Macdonald polynomials}
\label{se:mac}

In this section we present the first main property of the polynomials $P_{\lambda}(x_1,\dots,x_n;q,t;u,v)$ -- their reduction to Macdonald polynomials when $u=v=0$. We begin with some preliminary simplifying observations. 

\subsection{A simplification of the matrix product \eqref{linearform2}}
\label{se:simp1}

\begin{lem}
\label{aux-lem}
{\rm
Let $\mathcal{M} \in \bigotimes_{i,j=1}^{r} \mathcal{B}_i^{(j)}$ be any monomial in the expansion of $A_{\mu_1}(x_1) \cdots A_{\mu_n}(x_n)$, where the operators $A_{\mu_i}(x_i)$ are given by \eqref{ZF-ops}. Then if the annihilation operator $\phi_i^{(j)}$ appears in $\mathcal{M}$, a creation operator $\phi_i^{\dagger(\ell)}$ must also be present in $\mathcal{M}$, for some $\ell > j$.
}
\end{lem}

\begin{proof}
If $\phi_i^{(j)}$ appears in $\mathcal{M}$, it must have come from component $(i',i)$ of an $L$-matrix $L^{(j)}(x_{a})$, for some $0 \leq i' \not= i \leq r$ and $1 \leq a \leq n$. Let us denote this component by $L^{(j)}_{i'i}(x_a)$. Given that this component is selected, the component $L^{(j+1)}_{ii''}(x_a)$ must also be present, for some $0 \leq i'' \leq r$. If $i''\not=i$, $L^{(j+1)}_{ii''}(x_a)$ gives rise to the boson $\phi_i^{\dagger(j+1)}$, and the result follows. If $i''=i$, the boson $\phi_i^{\dagger(j+1)}$ is not produced, but we can iterate this reasoning to the next $L$-matrix in the product. If we need to iterate all the way to the final $L$-matrix in the product, the component $L^{(r)}_{i0}(x_a)$ will arise. This produces the boson $\phi_i^{\dagger(r)}$, since we know that $i\not=0$.
\end{proof}

We analyse more closely the expression \eqref{linearform2} for $f_{\mu}$, focusing on its dependence on the algebras $\mathcal{B}_i^{(1)},\dots,\mathcal{B}_i^{(r)}$. It is helpful to construct a two-dimensional visualisation of the matrix product, showing only that part of $\rho$ which acts on the $i^{\rm th}$ boson families $\mathcal{B}_i^{(1)},\dots,\mathcal{B}_i^{(r)}$:
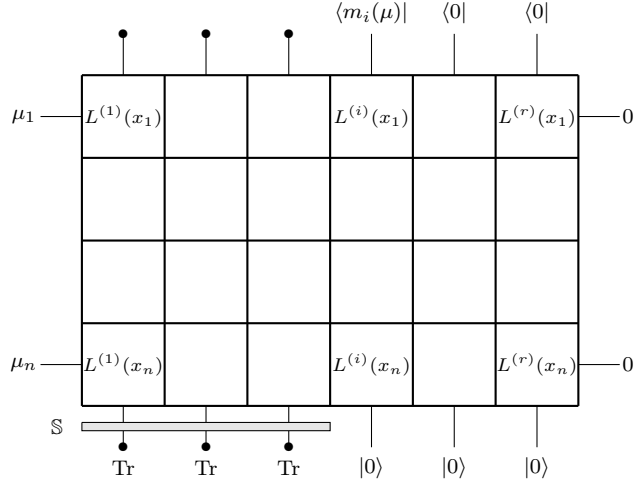
\begin{figure}[H]
\begin{tikzpicture}[scale=1.1,baseline=2cm]
\foreach\x in {0,...,6}{
\draw[thick] (\x,0) -- (\x,4);
}
\foreach\y in {0,...,4}{
\draw[thick] (0,\y) -- (6,\y);
}
\draw (0,3.5) -- (-0.5,3.5); \node[text centered] at (-0.7,3.5) {\scriptsize $\mu_1$};
\draw (0,0.5) -- (-0.5,0.5); \node[text centered] at (-0.7,0.5) {\scriptsize $\mu_n$};
\draw (6,3.5) -- (6.5,3.5); \node[text centered] at (6.6,3.5) {\scriptsize $0$};
\draw (6,0.5) -- (6.5,0.5); \node[text centered] at (6.6,0.5) {\scriptsize $0$};
\node[text centered] at (0.5,0.5) {\tiny $L^{(1)}(x_n)$};
\node[text centered] at (0.5,3.5) {\tiny $L^{(1)}(x_1)$};
\node[text centered] at (3.5,0.5) {\tiny $L^{(i)}(x_n)$};
\node[text centered] at (3.5,3.5) {\tiny $L^{(i)}(x_1)$};
\node[text centered] at (5.5,0.5) {\tiny $L^{(r)}(x_n)$};
\node[text centered] at (5.5,3.5) {\tiny $L^{(r)}(x_1)$};
\draw (0.5,-0.5) -- (0.5,0); \node[text centered] at (0.5,-0.75) {\scriptsize Tr}; \draw (0.5,4.5) -- (0.5,4); \node at (0.5,-0.5) {\scriptsize $\bullet$}; \node at (0.5,4.5) {\scriptsize $\bullet$};
\draw (1.5,-0.5) -- (1.5,0); \node[text centered] at (1.5,-0.75) {\scriptsize Tr}; \draw (1.5,4.5) -- (1.5,4); \node at (1.5,-0.5) {\scriptsize $\bullet$}; \node at (1.5,4.5) {\scriptsize $\bullet$};
\draw (2.5,-0.5) -- (2.5,0); \node[text centered] at (2.5,-0.75) {\scriptsize Tr}; \draw (2.5,4.5) -- (2.5,4); \node at (2.5,-0.5) {\scriptsize $\bullet$}; \node at (2.5,4.5) {\scriptsize $\bullet$};
\draw (3.5,-0.5) -- (3.5,0); \node[text centered] at (3.5,-0.75) {\scriptsize $\ket{0}$}; \node[text centered] at (3.5,4.75) {\scriptsize $\bra{m_i(\mu)}$}; \draw (3.5,4.5) -- (3.5,4);
\draw (4.5,-0.5) -- (4.5,0); \node[text centered] at (4.5,-0.75) {\scriptsize $\ket{0}$}; \node[text centered] at (4.5,4.75) {\scriptsize $\bra{0}$}; \draw (4.5,4.5) -- (4.5,4);
\draw (5.5,-0.5) -- (5.5,0); \node[text centered] at (5.5,-0.75) {\scriptsize $\ket{0}$}; \node[text centered] at (5.5,4.75) {\scriptsize $\bra{0}$}; \draw (5.5,4.5) -- (5.5,4);
\filldraw[fill=white!90!black] (0,-0.2) -- (3,-0.2) -- (3,-0.3) -- (0,-0.3) -- (0,-0.2);
\node[text centered] at (-0.3,-0.25) {\scriptsize $\mathbb{S}$};
\end{tikzpicture}
\captionsetup{labelformat=empty}
\caption{
Lattice representation of the matrix product \eqref{linearform2}, where for simplicity we only show the action of the linear form $\rho$ on the $i^{\rm th}$ $t$-boson families  $\mathcal{B}_i^{(j)}\ (j=1,\ldots,r)$, {\it i.e.}
$
\bra{m_i(\mu)}^{(i)}_i
\bra{0}^{(i+1)\dots (r)}_i
{\rm Tr}\left[
A_{\mu_1}(x_1)
\cdots
A_{\mu_n}(x_n)
\mathbb{S}
\right]^{(1)\dots (i-1)}_i
\ket{0}^{(i) \dots (r)}_i
$.
The full action of $\rho$ is obtained by taking a product over all $i=1,\dots,r$. 
}
\end{figure}

Each square of the lattice represents the contribution from a single $L$-matrix in the matrix product \eqref{linearform2}. In all columns $j \not= i$ we must have an equal number of annihilation and creation operators present, otherwise the resulting algebraic monomial will vanish under the action of $\rho$. In other words, we require that $\#(\phi_i^{(j)}) = \#(\phi_i^{\dagger(j)})$ in all columns $j \not= i$. 

This implies, in particular, that $\#(\phi_i^{(j)}) = \#(\phi_i^{\dagger(j)}) = 0$ in all columns $j>i$. Indeed, let us suppose that $j>i$ is the largest value such that $\#(\phi_i^{(j)}) = \#(\phi_i^{\dagger(j)}) \geq 1$, which means that necessarily $\#(\phi_i^{(\ell)}) = \#(\phi_i^{\dagger(\ell)}) = 0$ for all $\ell >j$. Given that 
$\phi_i^{(j)}$ appears, Lemma \ref{aux-lem} tells us that $\phi_i^{\dagger(\ell)}$ must appear for some $\ell > j$, leading to an immediate contradiction. We conclude that there is no value $j>i$ such that $\#(\phi_i^{(j)}) = \#(\phi_i^{\dagger(j)}) \geq 1$.

It follows that we are able to substitute $\phi_i^{(j)} \mapsto 0$, $\phi_i^{\dagger(j)} \mapsto 0$, $k_i^{(j)} \mapsto 1$ in \eqref{linearform2}, for all $i<j$, leaving it invariant. Such a substitution causes many entries of the participating $L$-matrices to vanish, greatly reducing the complexity of \eqref{linearform2}. 

\subsection{The $u=v=0$ case of \eqref{linearform2}}
\label{se:simp2}

The specialisation $u=v=0$ is a further, great simplification of \eqref{linearform2}. One easily shows that, after setting $\phi_i^{(j)} \mapsto 0$, $\phi_i^{\dagger(j)} \mapsto 0$, $k_i^{(j)} \mapsto 1$ for all $i<j$ and $u=v=0$, it is impossible for $\phi_i^{(i)}$ or $k_i^{(i)}$ to appear in \eqref{linearform2} for all $1\leq i \leq r$. It follows that $\phi_i^{\dagger(i)}$ appears exactly $m_i(\mu)$ times, and the expectation value $\bra{m_i(\mu)}^{(i)}_i (\phi_i^{\dagger(i)} )^{m_i(\mu)} \ket{0}_i^{(i)} = 1$ is effectively a common factor of $\eqref{linearform2}$.

Therefore, when $u=v=0$, we can additionally substitute $\phi_i^{\dagger(i)} \mapsto 1$ for all $1\leq i \leq r$ and omit the expectation value $\bra{m_i(\mu)}^{(i)}_i \cdots \ket{0}_i^{(i)}$ from the linear form $\rho$.

\subsection{Equivalence with matrix product formula of \cite{CantinidGW} at $u=v=0$}

An explicit matrix product expression for the Macdonald polynomials was given in \cite{CantinidGW}, using essentially the same method as in Sections \ref{sse:f} and \ref{se:P}. The $L$-matrix used in \cite{CantinidGW} is exactly the same as \eqref{eq:Lmat} with $u=v=0$, as can be easily checked.

There is however a change in notation between the presentation here and \cite{CantinidGW}. The construction of solutions to the ZF algebra in \cite{CantinidGW} makes use of a rank-reducing mechanism, whereby a monodromy matrix $T(x)=\tilde{L}^{(1)}(x)\cdots \tilde{L}^{(r)}(x)$ is written as in \eqref{monodromy} but with $\tilde{L}^{(i)}(x)$ an $(r-i+2)\times (r-i+1)$ rectangular matrix, whereas \eqref{monodromy} uses only the square $(r+1)$-dimensional $L$-matrix. 

The equivalence of the two approaches is based on the simplifications listed in Sections \ref{se:simp1} and \ref{se:simp2}. After performing the substitutions stated therein, one finds that the $L$-matrices in \eqref{linearform2} contain many redundant entries, which only give a vanishing contribution to $f_{\mu}(x_1,\dots,x_n;q,t;0,0)$. After suppressing these entries (which amounts to deleting rows and columns from the $L$-matrices), one arrives precisely at the rectangular $L$-matrices $\tilde{L}^{(i)}(x)$ of \cite{CantinidGW}. We have already seen, for instance, that \eqref{exampleA} reduces to \eqref{eq:ZFr=3} in the example of Section~\ref{se:polexample}. A general proof is elementary but tedious to explain in detail, so we will not elaborate further. 

We conclude that, at $u=v=0$, the matrix product \eqref{linearform2} recovers the family of non-symmetric polynomials $\{ f_{\mu}(x_1,\dots,x_n;q,t) \}_{\mu^{+} = \lambda}$ studied in \cite{CantinidGW}. Since the latter produce symmetric Macdonald polynomials $P_{\lambda}(x_1,\dots,x_n;q,t)$ by summing over all $\mu$ in the Weyl orbit of $\lambda$, we find that
\begin{align*}
P_{\lambda}(x_1,\dots,x_n;q,t;0,0)
=
P_{\lambda}(x_1,\dots,x_n;q,t)
\end{align*}
which we had set out to demonstrate.

\section{Specialisation to Borodin--Petrov rational symmetric functions}
\label{se:bor-pet}

In the recent papers \cite{Borodin,BorodinP}, Borodin and Petrov have introduced a rational, inhomogeneous generalisation of Hall--Littlewood polynomials, $\textsf{F}_{\lambda}(x_1,\dots,x_n;t;s)$. This generalisation is achieved via the inclusion of an additional parameter, $s$, which can be considered to parametrise the spin of a vertex model\footnote{We will not work in the full generality considered in \cite{BorodinP}, where a separate spin parameter and quantum impurity was introduced at each site of the lattice, preferring to focus on the functions $\textsf{F}_{\lambda}(x_1,\dots,x_n;t;s)$ as they were introduced in \cite{Borodin}.}. Indeed, at the special values $s=t^{-\ell/2}$, $\ell \in \mathbb{N}$, the functions $\textsf{F}_{\lambda}(x_1,\dots,x_n;t;s)$ reduce precisely to the wavefunctions of a spin-$\ell/2$ XXZ chain. The Hall--Littlewood polynomials themselves are recovered at $s=0$, which can accordingly be viewed as the limit of infinite spin. 

One of the main results of this paper is that the polynomials $P_{\lambda}(x_1,\dots,x_n;q,t;u,v)$ degenerate to the Borodin--Petrov family when $q=0$. Before proving this, we first point out some minor differences in convention that we use, in comparison with \cite{Borodin}. {\bf 1.} The members of the family $P_{\lambda}(x_1,\dots,x_n;q,t;u,v)$ are polynomials, so after taking $q=0$ we will obtain polynomials, not rational functions. This discrepancy can be cured by a harmless normalising factor depending on all $x$ variables. {\bf 2.} $P_{\lambda}(x_1,\dots,x_n;q,t;u,v)$ contains two deformation parameters, $u,v$, rather than a single $s$. After taking $q=0$ we obtain a function in $u,v$, which then reduces to the Borodin--Petrov case after setting $u=v=s$. {\bf 3.} To correctly perform the reduction, a trivial shift of the indexing partition is necessary. This will be explained in more detail in Remark \ref{conventions} below.

\subsection{Borodin--Petrov family}

Following \cite{Borodin}, we construct polynomials $F_{\lambda}(x_1,\dots,x_n;t;u,v)$ directly from the rank-1 integrable model of Example \ref{ex-rank1}. Let us again write the $L$-matrix, this time placing a superscript $j$ on the bosonic operators, to indicate a copy $\mathcal{B}^{(j)}$ of the $t$-boson algebra\footnote{Since this is a rank-1 model, there is only one family of bosons. Hence there is no need to place subscripts on bosonic operators.}:
\begin{align}
\label{rank1-L}
L^{(j)}(x)
=
\begin{pmatrix}
1-x u k & (1-u v k) \phi
\\
x \phid & x-v k
\end{pmatrix}^{(j)}.
\end{align}
A monodromy matrix is constructed by taking a product of these $L$-matrices, where $j$ ranges from $1$ up to $r$, the largest part of the partition that will subsequently interest us:
\begin{align}
\label{rank1-mon}
T(x)
=
L^{(1)}(x)
\cdots
L^{(r)}(x)
=
\begin{pmatrix}
T_{00}(x) & T_{01}(x)
\\
T_{10}(x) & T_{11}(x)
\end{pmatrix}.
\end{align}
\begin{defn}{\rm
Let $\lambda= 1^{m_1} \dots r^{m_r}$ be a partition with largest part $\lambda_1=r$, whose part multiplicities satisfy $\sum_{i=1}^{r} m_i(\lambda) \leq n$. We define symmetric polynomials $F_{\lambda}(x_1,\dots,x_n;t;u,v)$ as expectation values in the rank-1 model \eqref{rank1-L}, as follows:
\begin{align}
\label{F-defn}
F_{\lambda}(x_1,\dots,x_n;t;u,v)
:=
\bra{\lambda}
\mathcal{T}(x_1)
\dots
\mathcal{T}(x_n)
\ket{0},
\qquad
\bra{\lambda} = \bigotimes_{j=1}^{r} \bra{m_j(\lambda)}^{(j)},
\qquad
\ket{0} = \bigotimes_{j=1}^{r} \ket{0}^{(j)},
\end{align}
}
\end{defn}
\noindent
where $\mathcal{T}(x) = T_{00}(x) + T_{10}(x)$ is the sum of entries in the first column of \eqref{rank1-mon}.

\begin{remark}
\label{conventions}
{\rm
Up to differences in normalisation and a shift of the indexing partition $\lambda$, the polynomials \eqref{F-defn} are the same as those of Borodin--Petrov. Denoting the rational symmetric functions of \cite{Borodin} by 
$\textsf{F}_{\lambda}(x_1,\dots,x_n;t;s)$, the exact correspondence is given by
\begin{align}
\label{matching}
\textsf{F}_{\lambda}(x_1,\dots,x_n;t;s)
=
\left.
\frac{
F_{(\lambda+1)}(x_1,\dots,x_n;t;u,v)
}
{
\prod_{i=1}^{n}
x_i(1-x_i u)^{\lambda_1+1}
}
\right|_{u=v=s}
\end{align}
where $(\lambda+1)$ denotes the partition obtained from $\lambda$ by adding $1$ to every part. The factor of $(1-x_i u)^{\lambda_1+1}$ in the denominator is to account for the fact that \cite{Borodin} uses a rational normalisation of the $L$-matrix \eqref{rank1-mon}, whereas we adopt its polynomial normalisation. The appearance of $x_i$ in the denominator accounts for the slightly different gauge used in our solution of the rank-1 Yang--Baxter algebra, compared with \cite{Borodin}. To explain the shift in the partition, consider \eqref{F-defn} in the case $\sum_{i=1}^{r} m_i(\lambda) = n$. In that case, the $T_{00}(x_i)$ operators have a vanishing contribution to the expectation value in \eqref{F-defn}, and we can replace each $\mathcal{T}(x_i)$ by $T_{10}(x_i)$. The resulting expectation value then matches that of Borodin--Petrov after performing the shift $\lambda \mapsto (\lambda -1)$, which we are able to do, given that all parts of $\lambda$ are strictly positive in this case. The polynomials $F_{\lambda}(x_1,\dots,x_n;t;u,v)$ are therefore slightly more general than those studied in \cite{Borodin,BorodinP}, since they allow free boundary conditions at the left edge of the underlying lattice, as we shall shortly see.
}
\end{remark}

\begin{remark}{\rm
The functions $F_{\lambda}(x_1,\dots,x_n;t;u,v)$ can be expressed as an explicit sum over the symmetric group:
\begin{multline}
\label{F-sym}
F_{\lambda}(x_1,\dots,x_n;t;u,v)
=
\\
\frac{\prod_{i=1}^{n}(1-x_i u)^{\lambda_1}}{v_{\lambda}(t)}
\times
\sum_{\sigma \in S_n}
\sigma\left[
\prod_{1 \leq i<j \leq n}
\left(
\frac{x_i-tx_j}{x_i-x_j}
\right)
\prod_{i=1}^{n}
\left(
\frac{x_i-v}{1-x_i u}
\right)^{\lambda_i}
\left(
\frac{x_i}{x_i-v}
\right)^{\mathbbm{1}(\lambda_i > 0)}
\right],
\end{multline}
where $\mathbbm{1}(\cdot)$ denotes the indicator function, and $v_{\lambda}(t) = \prod_{i\geq 0} \prod_{j=1}^{m_i(\lambda)} (1-t^j)/(1-t)$ is a standard normalising factor from Hall--Littlewood theory \cite{MacdBook}. We omit the proof of this result, since we will not require it in our subsequent calculations. It can be proved either by simple modifications of the Bethe Ansatz approach in \cite{BorodinP}, or by the $F$-basis approach in \cite{WheelerZJ}. Equation \eqref{F-sym} allows easy comparison with the family introduced in \cite{Borodin}, when $u=v=s$. 
}
\end{remark}

\subsection{Lattice representation of $F_{\lambda}(x_1,\dots,x_n;t;u,v)$}
\label{uncoloured-model}

Another way of viewing $F_{\lambda}(x_1,\dots,x_n;t;u,v)$ is a partition function in an integrable lattice model. This is valuable not only for a more combinatorial understanding of $F_{\lambda}(x_1,\dots,x_n;t;u,v)$, but also for assigning to it a probabilistic interpretation \cite{BorodinP}. In the Hall--Littlewood case ($u=v=0$), this point of view has been well explored, see for example \cite{Korff,WheelerZJ}. Here we mostly follow the notation and conventions of \cite{Borodin,BorodinP}. One begins by representing the entries of the $L$-matrix \eqref{rank1-L} as vertices:
\begin{align}
\label{vertices}
\begin{array}{cccc}
\begin{tikzpicture}[scale=0.8]
\draw[dotted] (-1,0) -- (1,0);
\draw[dotted] (0,-1) -- (0,1);
\node[below] at (0,-1) {$m$};
\draw[thick,->] (-0.15,-1) -- (-0.15,1);
\draw[thick,->] (0,-1) -- (0,1);
\draw[thick,->] (0.15,-1) -- (0.15,1);
\node[above] at (0,1) {$m$};
\end{tikzpicture}
\quad\quad\quad\quad\quad
&
\begin{tikzpicture}[scale=0.8]
\draw[dotted] (-1,0) -- (1,0);
\draw[dotted] (0,-1) -- (0,1);
\node[below] at (0,-1) {$m$};
\draw[thick,->] (-0.15,-1) -- (-0.15,1);
\draw[thick,->] (0,-1) -- (0,1);
\draw[thick,->] (0.15,-1) -- (0.15,0) -- (1,0);
\node[above] at (0,1) {$m-1$};
\end{tikzpicture}
\quad\quad\quad\quad\quad
&
\begin{tikzpicture}[scale=0.8]
\draw[dotted] (-1,0) -- (1,0);
\draw[dotted] (0,-1) -- (0,1);
\node[below] at (0,-1) {$m$};
\draw[thick,->] (-1,0) -- (-0.15,0) -- (-0.15,1);
\draw[thick,->] (0,-1) -- (0,1);
\draw[thick,->] (0.15,-1) -- (0.15,1);
\node[above] at (0,1) {$m+1$};
\end{tikzpicture}
\quad\quad\quad\quad\quad
&
\begin{tikzpicture}[scale=0.8]
\draw[dotted] (-1,0) -- (1,0);
\draw[dotted] (0,-1) -- (0,1);
\node[below] at (0,-1) {$m$};
\draw[thick,->] (-1,0) -- (-0.15,0) -- (-0.15,1);
\draw[thick,->] (0,-1) -- (0,1);
\draw[thick,->] (0.15,-1) -- (0.15,0) -- (1,0);
\node[above] at (0,1) {$m$};
\end{tikzpicture}
\\
\bra{m} L_{00} \ket{m}
\quad\quad\quad\quad\quad
&
\bra{m-1} L_{01} \ket{m}
\quad\quad\quad\quad\quad
&
\bra{m+1} L_{10} \ket{m}
\quad\quad\quad\quad\quad
&
\bra{m} L_{11} \ket{m}
\\ \\
1-x u t^m
\quad\quad\quad\quad\quad
&
1-u v t^{m-1}
\quad\quad\quad\quad\quad
&
x(1-t^{m+1})
\quad\quad\quad\quad\quad
&
x-v t^m
\end{array}
\end{align}
Define the set $\mathcal{P}_n(\lambda)$, consisting of all possible configurations of $n$ paths on an $n \times \lambda_1$ lattice, subject to these boundary conditions: {\bf 1.} The bottom and right edges of the lattice are unoccupied, {\bf 2.} The left edges may be either occupied or unoccupied, {\bf 3.} The top edges are occupied according to the data $\{m_1,\dots,m_{\lambda_1}\}$. For example, in the case $n=4$ and $\lambda = (4,3,3,1)$, 
$\mathcal{P}_n(\lambda)$ is the set of all possible configurations on the lattice
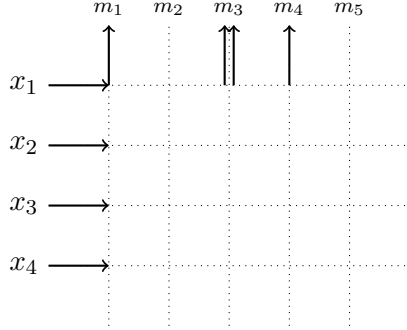
\begin{align*}
\begin{tikzpicture}[scale=0.8]
\foreach\y in {1,...,4}{
\draw[dotted] (0,\y) -- (6,\y);
\draw[thick,->] (0,\y) -- (1,\y);
\node[left] at (0,5-\y) {$x_\y$}; 
}
\foreach\x in {1,...,5}{
\draw[dotted] (\x,0) -- (\x,5);
\node[above] at (\x,5) {\scriptsize $m_\x$};
}
\draw[thick,->] (1,4) -- (1,5);
\draw[thick,->] (2.925,4) -- (2.925,5);
\draw[thick,->] (3.075,4) -- (3.075,5);
\draw[thick,->] (4,4) -- (4,5);
\end{tikzpicture}
\end{align*}
using the four types of vertices \eqref{vertices}. The Boltzmann weight of a single configuration $\mathcal{P}$ is the product of the Boltzmann weights of the constituent vertices, and denoted $W_{\mathcal{P}}(x_1,\dots,x_n;t;u,v)$. The expectation value \eqref{F-defn} can now be cast as a partition function of the set $\mathcal{P}_n(\lambda)$:
\begin{align}
\label{F-PF}
F_{\lambda}(x_1,\dots,x_n;t;u,v)
=
\sum_{\mathcal{P} \in \mathcal{P}_n(\lambda)}
W_{\mathcal{P}}(x_1,\dots,x_n;t;u,v).
\end{align}

\subsection{The $q=0$ case of the matrix product}

At $q=0$, the matrix product \eqref{symmP} greatly simplifies. This simplification is by virtue of the twist $\mathbb{S}$. One can easily see that all traces over $\phi^{(j)}_i$ bosons reduce to vacuum expectation values of the form $\langle 0|_i^{(j)} \cdots |0\rangle_i^{(j)}$, since all ``higher'' terms in the trace $\langle m|_i^{(j)} \cdots |m\rangle_i^{(j)}$ ($m \geq 1$) give rise to positive powers of $q$ and hence vanish. Letting $f_{\mu}(x_1,\dots,x_n;t;u,v)$ denote the polynomial $f_{\mu}(x_1,\dots,x_n;q,t;u,v)$ at $q=0$, we obtain
\begin{align}
\label{f_at_q=0}
f_{\mu}(x_1,\dots,x_n;t;u,v)
=
\textstyle{
\bigotimes_{i=1}^{r}
{\underbrace{\langle 0,\dots,0,m_i,0,\dots,0|}_{i^{\rm th}}}^{(i)}
}
A_{\mu_1}(x_1)
\dots
A_{\mu_n}(x_n)
\textstyle{
\bigotimes_{i=1}^{r}
|0,\dots,0 \rangle^{(i)}
}
\end{align}
where for each $1 \leq i \leq r$, $|0,\dots,0 \rangle^{(i)} \in \mathcal{F}^{(i)}_1 \otimes \cdots \otimes \mathcal{F}^{(i)}_r$ denotes a completely unoccupied state and $\langle 0,\dots,0,m_i,0,\dots,0|^{(i)} \in \mathcal{F}^{*(i)}_1 \otimes \cdots \otimes \mathcal{F}^{*(i)}_r$ is a dual state containing $m_i(\mu)$ particles of type $i$. Likewise, from \eqref{symmP} at $q=0$, the symmetrised polynomial $P_{\lambda}$ becomes
\begin{align}
\label{P_at_q=0}
P_{\lambda}(x_1,\dots,x_n;t;u,v)
=
\textstyle{
\bigotimes_{i=1}^{r}
{\underbrace{\langle 0,\dots,0,m_i,0,\dots,0|}_{i^{\rm th}}}^{(i)}
}
A(x_1)
\dots
A(x_n)
\textstyle{
\bigotimes_{i=1}^{r}
|0,\dots,0 \rangle^{(i)}.
}
\end{align}

\subsection{Lattice representation of $P_{\lambda}(x_1,\dots,x_n;t;u,v)$}
\label{coloured-model}

Proceeding along similar lines as above, one can represent \eqref{P_at_q=0} as a lattice partition function. It is first necessary to depict the entries of the $L$-matrix as vertices, and given that we are now working in a rank-$r$ model, we shall represent the different possible families by colouring our lattice paths. In what follows, the green lattice paths represent some family $i$, while blue paths represent another family $j$, where $1 \leq i<j \leq r$. Black lines indicate an arbitrary set of paths (of any colour) which propagate unchanged through the vertex:
\begin{align*}
\begin{array}{|c|c|c|}
\hline
\begin{tikzpicture}[scale=0.8]
\draw[dotted] (-1,0) -- (1,0);
\draw[dotted] (0,-1) -- (0,1);
\node[below] at (0,-1) {$\{m_1,\dots,m_r\}$};
\draw[thick,->] (-0.15,-1) -- (-0.15,1);
\draw[thick,->] (0,-1) -- (0,1);
\draw[thick,->] (0.15,-1) -- (0.15,1);
\node[above] at (0,1) {$\{m_1,\dots,m_r\}$};
\end{tikzpicture}
&
\begin{tikzpicture}[scale=0.8]
\draw[dotted] (-1,0) -- (1,0);
\draw[dotted] (0,-1) -- (0,1);
\node[below] at (0,-1) {$\{\dots,{\dg m_i},\dots\}$};
\draw[thick,->] (-0.15,-1) -- (-0.15,1);
\draw[thick,->] (0,-1) -- (0,1);
\draw[thick,->,green!65!black] (0.15,-1) -- (0.15,0) -- (1,0);
\node[above] at (0,1) {$\{\dots,{\dg m_i-1},\dots\}$};
\end{tikzpicture}
&
\begin{tikzpicture}[scale=0.8]
\draw[dotted] (-1,0) -- (1,0);
\draw[dotted] (0,-1) -- (0,1);
\node[below] at (0,-1) {$\{\dots,{\dg m_i},\dots\}$};
\draw[thick,->,green!65!black] (-1,0) -- (-0.15,0) -- (-0.15,1);
\draw[thick,->] (0,-1) -- (0,1);
\draw[thick,->] (0.15,-1) -- (0.15,1);
\node[above] at (0,1) {$\{\dots,{\dg m_i+1},\dots\}$};
\end{tikzpicture}
\\ \hline & & \\ 
1-x u t^{|m|}
&
1-u v t^{|m|-1}
&
x(1-t^{m_i+1}) t^{|m|_i}
\\ & & \\ \hline & & \\
\begin{tikzpicture}[scale=0.8]
\draw[dotted] (-1,0) -- (1,0);
\draw[dotted] (0,-1) -- (0,1);
\node[below] at (0,-1) {$\{m_1,\dots,m_r\}$};
\draw[thick,->,green!65!black] (-1,0) -- (-0.15,0) -- (-0.15,1);
\draw[thick,->] (0,-1) -- (0,1);
\draw[thick,->,green!65!black] (0.15,-1) -- (0.15,0) -- (1,0);
\node[above] at (0,1) {$\{m_1,\dots,m_r\}$};
\end{tikzpicture}
&
\begin{tikzpicture}[scale=0.8]
\draw[dotted] (-1,0) -- (1,0);
\draw[dotted] (0,-1) -- (0,1);
\node[below] at (0,-1) {$\{\dots,{\dg m_i},\dots,{\bl m_j},\dots\}$};
\draw[thick,->,blue!65!black] (-1,0) -- (-0.15,0) -- (-0.15,1);
\draw[thick,->] (0,-1) -- (0,1);
\draw[thick,->,green!65!black] (0.15,-1) -- (0.15,0) -- (1,0);
\node[above] at (0,1) {$\{\dots,{\dg m_i-1},\dots,{\bl m_j+1},\dots\}$};
\end{tikzpicture}
&
\begin{tikzpicture}[scale=0.8]
\draw[dotted] (-1,0) -- (1,0);
\draw[dotted] (0,-1) -- (0,1);
\node[below] at (0,-1) {$\{\dots,{\dg m_i},\dots,{\bl m_j},\dots\}$};
\draw[thick,->,green!65!black] (-1,0) -- (-0.15,0) -- (-0.15,1);
\draw[thick,->] (0,-1) -- (0,1);
\draw[thick,->,blue!65!black] (0.15,-1) -- (0.15,0) -- (1,0);
\node[above] at (0,1) {$\{\dots,{\dg m_i+1},\dots,{\bl m_j-1},\dots\}$};
\end{tikzpicture}
\\ \hline & & \\
(x-v t^{m_i}) t^{|m|_i}
&
x (1-t^{m_j+1}) t^{|m|_j}
&
v (1-t^{m_i+1}) t^{|m|_i-1}
\\ & & \\ \hline
\end{array}
\end{align*}
where in all cases $|m|:=\sum_{\ell=1}^{r} m_\ell$ and $|m|_i := \sum_{\ell = i+1}^{r} m_{\ell}$. 

Introduce the set $\mathcal{C}_n(\lambda)$, consisting of all possible configurations of coloured paths on an $n \times \lambda_1$ lattice, subject to the following boundary conditions: {\bf 1.} The bottom and right edges of the lattice are unoccupied, {\bf 2.} Each left edge may be occupied by any coloured path or unoccupied, {\bf 3.} The top edge in the $i^{\rm th}$ column is occupied by $m_i(\lambda)$ paths of colour $i$, and no other paths. For example, in the case $n=4$ and $\lambda = (4,3,3,1)$, $\mathcal{C}_n(\lambda)$ is the set of all configurations on the lattice
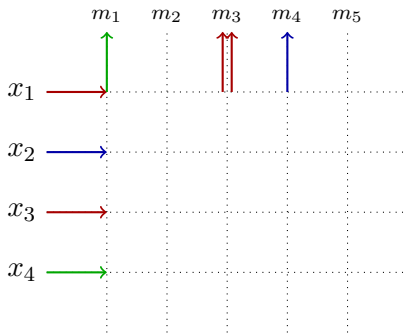
\begin{align*}
\begin{tikzpicture}[scale=0.8]
\foreach\y in {1,...,4}{
\draw[dotted] (-1,\y) -- (5,\y);
\node[left] at (-1,5-\y) {$x_\y$};
}
\draw[thick,red!65!black,->] (-1,4) -- (0,4);
\draw[thick,blue!65!black,->] (-1,3) -- (0,3);
\draw[thick,red!65!black,->] (-1,2) -- (0,2);
\draw[thick,green!65!black,->] (-1,1) -- (0,1);
\foreach\x in {1,...,5}{
\draw[dotted] (\x-1,0) -- (\x-1,5);
\node[above] at (\x-1,5) {\scriptsize $m_\x$};
}
\draw[thick,green!65!black,->] (0,4) -- (0,5);
\draw[thick,red!65!black,->] (1.925,4) -- (1.925,5);
\draw[thick,red!65!black,->] (2.075,4) -- (2.075,5);
\draw[thick,blue!65!black,->] (3,4) -- (3,5);
\end{tikzpicture}
\end{align*}
in addition to all configurations which are possible by permuting the colours at the left edge. The Boltzmann weight of a configuration $\mathcal{C}$ is, once again, the product of the Boltzmann weights of its constituent vertices, and denoted $\widetilde{W}_{\mathcal{C}}(x_1,\dots,x_n;t;u,v)$. The lattice representation of \eqref{P_at_q=0} is then
\begin{align}
\label{P-PF}
P_{\lambda}(x_1,\dots,x_n;t;u,v)
=
\sum_{\mathcal{C} \in \mathcal{C}_n(\lambda)}
\widetilde{W}_{\mathcal{C}}(x_1,\dots,x_n;t;u,v).
\end{align}

\begin{defn}{\rm
Let $\mathcal{C}$ be a coloured path configuration in $C_n(\lambda)$. The \textit{black-and-white projection} of 
$\mathcal{C}$, denoted $\mathcal{C}^{*}$, is the profile traced out by the paths in $\mathcal{C}$. In other words, $\mathcal{C}^{*}$ is obtained from $\mathcal{C}$ by colouring all paths black.
}
\end{defn}

\subsection{Equivalence of polynomials}

At this stage, one notices the strong similarity between the form of \eqref{F-defn} and \eqref{P_at_q=0}. Aside from the fact that \eqref{F-defn} is expressible in terms of rank-1 $L$-matrices, while \eqref{P_at_q=0} makes use of rank-$r$ $L$-matrices, it is conceivable that the two expressions are related. In fact, they are equal:
\begin{prop}
\label{main-prop}
{\rm 
For all partitions $\lambda$, we have
\begin{align}
\label{q=0_result}
F_{\lambda}(x_1,\dots,x_n;t;u,v)
=
P_{\lambda}(x_1,\dots,x_n;t;u,v).
\end{align}
}
\end{prop}
There is an even stronger result, related to the combinatorial interpretation of \eqref{q=0_result}. On the left hand side, we have a partition function \eqref{F-PF} whose configurations are lattice paths which propagate SW $\longrightarrow$ NE. On the right hand side, we have a partition function \eqref{P-PF} featuring coloured lattice path configurations, propagating in the same direction. It is natural, then, to search for a correspondence which identifies a single term in the partition function on the left hand side with multiple terms on the right hand side. Such a correspondence exists: 
\begin{prop}
\label{prop-colour}{\rm
Let $\mathcal{P}$ be a configuration of lattice paths in the sum \eqref{F-PF}, and $W_{\mathcal{P}}$ its corresponding Boltzmann weight. Similarly let $\mathcal{C}$ be a configuration of coloured lattice paths in the sum \eqref{P-PF}, $\widetilde{W}_{\mathcal{C}}$ its Boltzmann weight, and $\mathcal{C}^{*}$ the black-and-white projection of $\mathcal{C}$. Then
\begin{align}
\label{colour-ind}
W_{\mathcal{P}}(x_1,\dots,x_n;t;u,v)
=
\sum_{\substack{\mathcal{C} \in \mathcal{C}_n(\lambda) \\ \mathcal{C}^* = \mathcal{P}}}
\widetilde{W}_{\mathcal{C}}(x_1,\dots,x_n;t;u,v).
\end{align}
}
\end{prop}

\begin{remark}{\rm
Equation \eqref{colour-ind} is very analogous to identities obtained in \cite{FodaW}. The ``colour-independence'' property was used in \cite{FodaW} to show that certain partition functions in $sl(n)$ vertex models are in fact equal to $sl(2)$ counterparts, much as in the situation at hand.
}
\end{remark}

The rest of this section will be devoted to the proof of \eqref{colour-ind}. Proposition \ref{main-prop} follows as an immediate corollary, by summing \eqref{colour-ind} over all path profiles $\mathcal{P}$. The first step is to prove the following, very powerful theorem:
\begin{thm}
\label{amazing}
{\rm
Let $L_{ij}(x)$ denote component $(i,j)$ of the rank-$r$ $L$-matrix, where $i,j \in \{0,1,\dots,r\}$. Let $|m_1,\dots,m_r\rangle$ denote a generic bosonic state in $\mathcal{F}_1 \otimes \cdots \otimes \mathcal{F}_r$, and define
\begin{align*}
\Ket{M}
=
\sum_{\substack{\{m_1,\dots,m_r\} \\ |m|=M}}
|m_1,\dots,m_r\rangle,
\qquad
\text{where}\ 
|m| = \sum_{i=1}^{r} m_i.
\end{align*}
In particular, $\Ket{0} = |0,\dots,0\rangle$. The following four relations are valid for any $M \geq 0$:
\begin{align}
\label{1}
L_{00}(x) \Ket{M} 
&=
(1-x u t^M)
\Ket{M},
\\
\label{2}
L_{0j}(x) 
\Ket{M}
&=
(1- uv t^{M-1})
\Ket{M-1},
\quad
\forall\ j \in \{1,\dots,r\},
\\
\label{3}
\sum_{i=1}^{r}
L_{i0}(x)
\Ket{M}
&=
x
(1-t^{M+1})
\Ket{M+1},
\\
\label{4}
\sum_{i=1}^{r}
L_{ij}(x)
\Ket{M}
&=
(x-v t^M)
\Ket{M},
\quad
\forall\ j \in \{1,\dots,r\}.
\end{align}
}
\end{thm}

\begin{proof}{\rm
The proof is by direct computation, using the explicit form of the entries of $L(x)$. The relations \eqref{1} and \eqref{2} are immediate, since
\begin{align*}
L_{00}(x)
\Ket{M}
=
\sum_{\substack{\{m_1,\dots,m_r\} \\ |m|=M}}
(1- x u k_1\dots k_r )
|m_1,\dots,m_r\rangle
=
\sum_{\substack{\{m_1,\dots,m_r\} \\ |m|=M}}
(1- x u t^{|m|} )
|m_1,\dots,m_r\rangle,
\end{align*}
and
\begin{multline*}
L_{0j}(x) 
\Ket{M}
=
\sum_{\substack{\{m_1,\dots,m_r\} \\ |m|=M}}
(1- u v k_1\dots k_r)
\phi_j
|m_1,\dots,m_r\rangle
\\
=
\sum_{\substack{\{m_1,\dots,m_r\} \\ |m|=M-1}}
(1- u v k_1\dots k_r)
|m_1,\dots,m_r\rangle
=
\sum_{\substack{\{m_1,\dots,m_r\} \\ |m|=M-1}}
(1- u v t^{|m|})
|m_1,\dots,m_r\rangle.
\end{multline*}
The relation \eqref{3} is slightly more complicated. We find that
\begin{multline*}
L_{i0}(x)
\Ket{M}
=
\sum_{\substack{\{m_1,\dots,m_r\} \\ |m|=M}}
x
k_{r} \dots k_{i+1}
\phid_i
|m_1,\dots,m_r\rangle
\\
=
\sum_{\substack{\{m_1,\dots,m_r\} \\ |m|=M}}
x
(1-t^{m_i+1}) t^{m_{i+1}} \dots t^{m_r}
|\dots,m_i+1,\dots\rangle
=
\sum_{\substack{\{m_1,\dots,m_r\} \\ |m|=M+1}}
x
(1-t^{m_i}) t^{m_{i+1}} \dots t^{m_r}
|m_1,\dots,m_r\rangle,
\end{multline*}
where we used the fact that $(1-t^{m_i})=0$ if $m_i=0$ to write the final summation. Summing over all $1 \leq i \leq r$ produces a telescoping sum:
\begin{align*}
\sum_{i=1}^{r}
L_{i0}(x)
\Ket{M}
=
\sum_{\substack{\{m_1,\dots,m_r\} \\ |m|=M+1}}
x
(1-t^{m_1} \dots t^{m_r})
|m_1,\dots,m_r\rangle
=
\sum_{\substack{\{m_1,\dots,m_r\} \\ |m|=M+1}}
x
(1-t^{|m|})
|m_1,\dots,m_r\rangle.
\end{align*}
The relation \eqref{4} requires the most work. Using the explicit form of the entries 
$L_{ij}(x)$, we have
\begin{multline*}
\sum_{i=1}^{r}
L_{ij}(x)
\Ket{M}
\\
=
\sum_{\substack{\{m_1,\dots,m_r\} \\ |m|=M}}
\left(
\sum_{i=1}^{j-1}
(v
k_r \dots k_{i+1}
\phid_i \phi_j)
+
(x-v k_j) k_{j+1} \dots k_r
+
\sum_{i=j+1}^{r}
(x
k_r \dots k_{i+1} \phid_i \phi_j)
\right)
|m_1,\dots,m_r\rangle,
\end{multline*}
and after acting with each operator on the bosonic state vector, we obtain
\begin{multline*}
\sum_{i=1}^{r}
L_{ij}(x)
\Ket{M}
=
\sum_{\substack{\{m_1,\dots,m_r\} \\ |m|=M}}
\left(
\sum_{i=1}^{j-1}
(v t^{m_r} \dots t^{m_{i+1}}
(1-t^{m_i}))
\right.
\\
\left.
-
v t^{m_r} \dots t^{m_{j}}
+
x t^{m_r} \dots t^{m_{j+1}}
+
\sum_{i=j+1}^{r}
(x t^{m_r} \dots t^{m_{i+1}}
(1-t^{m_i}))
\right)
|m_1,\dots,m_r\rangle.
\end{multline*}
Similarly to above, the internal sums are telescoping, and this simplifies to
\begin{align*}
\sum_{i=1}^{r}
L_{ij}(x)
\Ket{M}
=
\sum_{\substack{\{m_1,\dots,m_r\} \\ |m|=M}}
(x-v t^{m_1}\dots t^{m_r})
|m_1,\dots,m_r\rangle.
\end{align*}
}
\end{proof}

\begin{cor}{\rm 
Equation \eqref{colour-ind} holds.
}
\end{cor}

\begin{proof}
The relations \eqref{1}--\eqref{4} have a simple combinatorial meaning. Consider a single vertex in the model of Section \ref{coloured-model}, and do the following:
\begin{enumerate}
\item Take the left edge to be {\bf (a)} unoccupied or {\bf (b)} occupied, in which case a sum is taken over all possible colours,
\item Sum the bottom edge over all possible ways of colouring $M$ particles,
\item Fix the right edge to be {\bf (c)} unoccupied or {\bf (d)} occupied, in which case {\it any} colour may be chosen,
\item Fix the top edge to {\it any} set of $N$ coloured particles.
\end{enumerate}
Theorem \ref{amazing} says that this sum is equal to the Boltzmann weight of the corresponding vertex in the uncoloured model (the model in Section \ref{uncoloured-model}). In other words, it is equal to the weight of the vertex whose
\begin{enumerate}
\item Left edge is {\bf (a)} unoccupied or {\bf (b)} occupied, respectively,
\item Bottom edge contains $M$ particles,
\item Right edge is {\bf (c)} unoccupied or {\bf (d)} occupied, respectively,
\item Top edge contains $N$ coloured particles.
\end{enumerate}
Furthermore, this result can be readily iterated over any rectangular lattice of vertices in the coloured model, so long as each external left and bottom edge is summed as we have described. The partition function \eqref{P-PF} satisfies these criteria: each left edge is summed over all possible unoccupied/occupied states, while each bottom edge is summed (trivially) over all ways of colouring $0$ particles. The result \eqref{colour-ind} follows immediately.

\end{proof}

\subsection{An example of Proposition \ref{prop-colour}}

To illustrate more clearly the simple meaning of \eqref{colour-ind}, we give here an explicit example for the running case $n=4$, $\lambda = (4,3,3,1)$. A permissible path configuration $\mathcal{P} \in \mathcal{P}_n(\lambda)$ in that case would be
\begin{align*}
\begin{array}{|c|c|}
\hline
\mathcal{P}
&
W_{\mathcal{P}}
\\
\hline
\begin{tikzpicture}[scale=0.6]
\foreach\y in {1,...,4}{
\draw[dotted] (0,\y) -- (5,\y);
\draw[thick,->] (0,\y) -- (1,\y);
\node[left] at (0,5-\y) {$x_\y$};
}
\foreach\x in {1,...,4}{
\draw[dotted] (\x,0) -- (\x,5);
\node[above] at (\x,5) {\scriptsize $m_\x$};
}
\draw[thick] (0,4) -- (1,4); \draw[thick,->] (1,4) -- (1,5);
\draw[thick] (0,3) -- (1,3) -- (1,4) -- (2.85,4); \draw[thick,->] (2.85,4) -- (2.85,5);
\draw[thick] (0,2) -- (2,2) -- (2,3) -- (3,3) -- (3,4); \draw[thick,->] (3,4) -- (3,5);
\draw[thick] (0,1) -- (3.15,1) -- (3.15,4) -- (4,4); \draw[thick,->] (4,4) -- (4,5);
\end{tikzpicture}
&
\begin{tikzpicture}[scale=0.6,baseline=-0.1cm]
\node at (0,4.15) {\scriptsize{$(x_1-vt)(x_1-v)(x_1-vt^2)x_1(1-t)$}};
\node at (0,3.05) {\scriptsize{$x_2(1-t)(1-uv) x_2(1-t^2)$}};
\node at (0,1.95) {\scriptsize{$(x_3-v)x_3(1-t)(1-u t x_3)$}};
\node at (0,0.85) {\scriptsize{$(x_4-v)(x_4-v) x_4(1-t)$}};
\end{tikzpicture}
\\
\hline
\end{array}
\end{align*}
where we have written the corresponding Boltzmann weight $W_{\mathcal{P}}$ alongside. Proposition \ref{prop-colour} states that the same result will be obtained if we sum over all coloured configurations $\mathcal{C}$ in the higher-rank model, whose profile $\mathcal{C}^{*}$ matches $\mathcal{P}$. There are six such configurations:
\begin{align*}
\begin{array}{|c|c|}
\hline
\mathcal{C}
&
\widetilde{W}_{\mathcal{C}}
\\
\hline
\begin{tikzpicture}[scale=0.6]
\foreach\y in {1,...,4}{
\draw[dotted] (0,\y) -- (5,\y);
\node[left] at (0,5-\y) {$x_\y$};
}
\draw[thick,green!65!black,->] (0,4) -- (1,4); \draw[thick,green!65!black] (0,4) -- (1,4);
\draw[thick,red!65!black,->] (0,3) -- (1,3); \draw[thick,red!65!black] (0,3) -- (1,3) -- (1,4) -- (2.85,4);
\draw[thick,red!65!black,->] (0,2) -- (1,2);  \draw[thick,red!65!black] (0,2) -- (2,2) -- (2,3) -- (3,3) -- (3,4);
\draw[thick,blue!65!black,->] (0,1) -- (1,1); \draw[thick,blue!65!black] (0,1) -- (3.15,1) -- (3.15,4) -- (4,4);
\foreach\x in {1,...,4}{
\draw[dotted] (\x,0) -- (\x,5);
\node[above] at (\x,5) {\scriptsize $m_\x$};
}
\draw[thick,green!65!black,->] (1,4) -- (1,5); 
\draw[thick,red!65!black,->] (2.85,4) -- (2.85,5);
\draw[thick,red!65!black,->] (3,4) -- (3,5);
\draw[thick,blue!65!black,->] (4,4) -- (4,5);
\end{tikzpicture}
&
\begin{tikzpicture}[scale=0.6,baseline=-0.1cm]
\node at (0,4.15) {\scriptsize{$ v(1-t)(x_1-v) v(1-t^2)x_1(1-t)$}};
\node at (0,3.05) {\scriptsize{$ x_2(1-t)(1-uv) x_2(1-t)t$}};
\node at (0,1.95) {\scriptsize{$(x_3-v)x_3(1-t)(1-u t x_3)$}};
\node at (0,0.85) {\scriptsize{$(x_4-v)(x_4-v) x_4(1-t)$}};
\end{tikzpicture}
\\
\hline
\begin{tikzpicture}[scale=0.6]
\foreach\y in {1,...,4}{
\draw[dotted] (0,\y) -- (5,\y);
\node[left] at (0,5-\y) {$x_\y$};
}
\draw[thick,green!65!black,->] (0,4) -- (1,4); \draw[thick,green!65!black] (0,4) -- (1,4);
\draw[thick,red!65!black,->] (0,3) -- (1,3); \draw[thick,red!65!black] (0,3) -- (1,3) -- (1,4) -- (2.85,4);
\draw[thick,blue!65!black,->] (0,2) -- (1,2); \draw[thick,blue!65!black] (0,2) -- (2,2) -- (2,3) -- (3.15,3) -- (3.15,4) -- (4,4);
\draw[thick,red!65!black,->] (0,1) -- (1,1); \draw[thick,red!65!black] (0,1) -- (3,1) -- (3,4);
\foreach\x in {1,...,4}{
\draw[dotted] (\x,0) -- (\x,5);
\node[above] at (\x,5) {\scriptsize $m_\x$};
}
\draw[thick,green!65!black,->] (1,4) -- (1,5);
\draw[thick,red!65!black,->] (2.85,4) -- (2.85,5);
\draw[thick,red!65!black,->] (3,4) -- (3,5);
\draw[thick,blue!65!black,->] (4,4) -- (4,5);
\end{tikzpicture}
&
\begin{tikzpicture}[scale=0.6,baseline=-0.1cm]
\node at (0,4.15) {\scriptsize{$ v(1-t)(x_1-v) v(1-t^2) x_1(1-t)$}};
\node at (0,3.05) {\scriptsize{$ x_2(1-t)(1-uv) x_2(1-t)$}};
\node at (0,1.95) {\scriptsize{$ (x_3-v) x_3(1-t)(1-u t x_3)$}};
\node at (0,0.85) {\scriptsize{$ (x_4-v)(x_4-v) x_4(1-t)$}};
\end{tikzpicture}
\\
\hline
\begin{tikzpicture}[scale=0.6]
\foreach\y in {1,...,4}{
\draw[dotted] (0,\y) -- (5,\y);
\node[left] at (0,5-\y) {$x_\y$};
}
\draw[thick,green!65!black,->] (0,4) -- (1,4); \draw[thick,green!65!black] (0,4) -- (1,4);
\draw[thick,blue!65!black,->] (0,3) -- (1,3); \draw[thick,blue!65!black] (0,3) -- (1,3) -- (1,4) -- (4,4);
\draw[thick,red!65!black,->] (0,2) -- (1,2);  \draw[thick,red!65!black] (0,2) -- (2,2) -- (2,3) -- (2.85,3) -- (2.85,4);
\draw[thick,red!65!black,->] (0,1) -- (1,1); \draw[thick,red!65!black] (0,1) -- (3,1) -- (3,4);
\foreach\x in {1,...,4}{
\draw[dotted] (\x,0) -- (\x,5);
\node[above] at (\x,5) {\scriptsize $m_\x$};
}
\draw[thick,green!65!black,->] (1,4) -- (1,5);
\draw[thick,red!65!black,->] (2.85,4) -- (2.85,5);
\draw[thick,red!65!black,->] (3,4) -- (3,5);
\draw[thick,blue!65!black,->] (4,4) -- (4,5);
\end{tikzpicture}
&
\begin{tikzpicture}[scale=0.6,baseline=-0.1cm]
\node at (0,4.15) {\scriptsize{$ v(1-t)(x_1-v) (x_1-v) x_1(1-t)$}};
\node at (0,3.05) {\scriptsize{$ x_2(1-t)(1-uv) x_2(1-t^2)$}};
\node at (0,1.95) {\scriptsize{$ (x_3-v) x_3(1-t)(1-u t x_3)$}};
\node at (0,0.85) {\scriptsize{$ (x_4-v)(x_4-v) x_4(1-t)$}};
\end{tikzpicture}
\\
\hline
\end{array}
%%%%%%%%%%%
\begin{array}{|c|c|}
\hline
\mathcal{C}
&
\widetilde{W}_{\mathcal{C}}
\\
\hline
\begin{tikzpicture}[scale=0.6]
\foreach\y in {1,...,4}{
\draw[dotted] (0,\y) -- (5,\y);
\node[left] at (0,5-\y) {$x_\y$};
}
\draw[thick,red!65!black,->] (0,4) -- (1,4); \draw[thick,red!65!black] (0,4) -- (2.85,4);
\draw[thick,green!65!black,->] (0,3) -- (1,3); \draw[thick,green!65!black] (0,3) -- (1,3) -- (1,4);
\draw[thick,red!65!black,->] (0,2) -- (1,2);  \draw[thick,red!65!black] (0,2) -- (2,2) -- (2,3) -- (3,3) -- (3,4); 
\draw[thick,blue!65!black,->] (0,1) -- (1,1); \draw[thick,blue!65!black] (0,1) -- (3.15,1) -- (3.15,4) -- (4,4);
\foreach\x in {1,...,4}{
\draw[dotted] (\x,0) -- (\x,5);
\node[above] at (\x,5) {\scriptsize $m_\x$};
}
\draw[thick,green!65!black,->] (1,4) -- (1,5);
\draw[thick,red!65!black,->] (2.85,4) -- (2.85,5);
\draw[thick,red!65!black,->] (3,4) -- (3,5);
\draw[thick,blue!65!black,->] (4,4) -- (4,5);
\end{tikzpicture}
&
\begin{tikzpicture}[scale=0.6,baseline=-0.1cm]
\node at (0,4.15) {\scriptsize{$ (x_1-v)(x_1-v) v(1-t^2) x_1(1-t)$}};
\node at (0,3.05) {\scriptsize{$ x_2(1-t) (1-uv) x_2(1-t)t$}};
\node at (0,1.95) {\scriptsize{$ (x_3-v) x_3(1-t)(1-u t x_3)$}};
\node at (0,0.85) {\scriptsize{$ (x_4-v) (x_4-v) x_4(1-t)$}};
\end{tikzpicture}
\\
\hline
\begin{tikzpicture}[scale=0.6]
\foreach\y in {1,...,4}{
\draw[dotted] (0,\y) -- (5,\y);
\node[left] at (0,5-\y) {$x_\y$};
}
\draw[thick,red!65!black,->] (0,4) -- (1,4); \draw[thick,red!65!black] (0,4) -- (2.85,4);
\draw[thick,green!65!black,->] (0,3) -- (1,3); \draw[thick,green!65!black] (0,3) -- (1,3) -- (1,4);
\draw[thick,blue!65!black,->] (0,2) -- (1,2); \draw[thick,blue!65!black] (0,2) -- (2,2) -- (2,3) -- (3.15,3) -- (3.15,4) -- (4,4);
\draw[thick,red!65!black,->] (0,1) -- (1,1); \draw[thick,red!65!black] (0,1) -- (3,1) -- (3,4);
\foreach\x in {1,...,4}{
\draw[dotted] (\x,0) -- (\x,5);
\node[above] at (\x,5) {\scriptsize $m_\x$};
}
\draw[thick,green!65!black,->] (1,4) -- (1,5);
\draw[thick,red!65!black,->] (2.85,4) -- (2.85,5);
\draw[thick,red!65!black,->] (3,4) -- (3,5);
\draw[thick,blue!65!black,->] (4,4) -- (4,5);
\end{tikzpicture}
&
\begin{tikzpicture}[scale=0.6,baseline=-0.1cm]
\node at (0,4.15) {\scriptsize{$ (x_1-v) (x_1-v) v(1-t^2) x_1(1-t) $}};
\node at (0,3.05) {\scriptsize{$ x_2(1-t)(1-uv) x_2(1-t)$}};
\node at (0,1.95) {\scriptsize{$ (x_3-v) x_3(1-t)(1-u t x_3)$}};
\node at (0,0.85) {\scriptsize{$ (x_4-v) (x_4-v) x_4(1-t)$}};
\end{tikzpicture}
\\
\hline
\begin{tikzpicture}[scale=0.6]
\foreach\y in {1,...,4}{
\draw[dotted] (0,\y) -- (5,\y);
\node[left] at (0,5-\y) {$x_\y$};
}
\draw[thick,blue!65!black,->] (0,4) -- (1,4); \draw[thick,blue!65!black] (0,4) -- (4,4);
\draw[thick,green!65!black,->] (0,3) -- (1,3); \draw[thick,green!65!black] (0,3) -- (1,3) -- (1,4);
\draw[thick,red!65!black,->] (0,2) -- (1,2);  \draw[thick,red!65!black] (0,2) -- (2,2) -- (2,3) -- (2.85,3) -- (2.85,4);
\draw[thick,red!65!black,->] (0,1) -- (1,1); \draw[thick,red!65!black] (0,1) -- (3,1) -- (3,4);
\foreach\x in {1,...,4}{
\draw[dotted] (\x,0) -- (\x,5);
\node[above] at (\x,5) {\scriptsize $m_\x$};
}
\draw[thick,green!65!black,->] (1,4) -- (1,5);
\draw[thick,red!65!black,->] (2.85,4) -- (2.85,5);
\draw[thick,red!65!black,->] (3,4) -- (3,5);
\draw[thick,blue!65!black,->] (4,4) -- (4,5);
\end{tikzpicture}
&
\begin{tikzpicture}[scale=0.6,baseline=-0.1cm]
\node at (0,4.15) {\scriptsize{$ (x_1-v) (x_1-v)(x_1-v) x_1(1-t)$}};
\node at (0,3.05) {\scriptsize{$ x_2(1-t)(1-uv) x_2(1-t^2)$}};
\node at (0,1.95) {\scriptsize{$ (x_3-v) x_3(1-t) (1-u t x_3)$}};
\node at (0,0.85) {\scriptsize{$ (x_4-v)(x_4-v) x_4(1-t)$}};
\end{tikzpicture}
\\
\hline
\end{array}
\end{align*}
Summing the six weights $\widetilde{W}_C$, we recover $W_{\mathcal{P}}$. 

\section{Construction of the $L$-matrix}
\label{proofrll}
In this section we will show how to construct the integrability objects $R(x,y)$ and $L(x,y)$ of the quantum group $\hat{\mathfrak{g}}=U_{\tau}(A_r^{(1)})$ using the formula of Jimbo \cite{Jimbo86a}. As a result the statement of Theorem \ref{rll} will follow.

The quantum affine Lie algebra $\hat{\mathfrak{g}}=U_{\tau}(A_r^{(1)})$ possesses the universal $R$-matrix $\mathcal{R}\in \hat{\mathfrak{g}}\otimes \hat{\mathfrak{g}}$ \cite{Dr87,Jimbo86b} which satisfies the Yang--Baxter equation 
\begin{align}\label{URYB}
\mathcal{R}_{1,2}\mathcal{R}_{1,3}\mathcal{R}_{2,3}=
\mathcal{R}_{2,3}\mathcal{R}_{1,3}\mathcal{R}_{1,2}.
\end{align}
This is an identity in $\hat{\mathfrak{g}}\otimes \hat{\mathfrak{g}}\otimes \hat{\mathfrak{g}}$ and the indices of $\mathcal{R}$  specify which two copies of $\hat{\mathfrak{g}}$ it belongs to. Let $\hat{\rho}_x$ be the fundamental representation of $\hat{\mathfrak{g}}$ on the space $V$ obtained via Jimbo homomorphism and $\pi_x$ be the homomorphism between $\hat{\mathfrak{g}}$ and its finite version $\mathfrak{g}$, $\pi_x: \hat{\mathfrak{g}} \rightarrow  \mathfrak{g} \otimes \mathbb{C}[x,x^{-1}]$, where $\mathfrak{g}=U'_{\tau}(sl_{r+1})$ is a certain refinement of $U_{\tau}(sl_{r+1})$  (the details are given below). Then applying $\hat{\rho}_x\otimes \hat{\rho}_y \otimes \pi_z$  to (\ref{URYB}) we get
\begin{align}\label{URLL}
R_{1,2}(x,y)\mathcal{L}_{1,3}(x,z)\mathcal{L}_{2,3}(y,z)=
\mathcal{L}_{2,3}(y,z)\mathcal{L}_{1,3}(x,z)R_{1,2}(x,y),
\end{align}
where $R(x,y)\in V\otimes V$ is the well known $R$-matrix of $\hat{\mathfrak{g}}$ and $\mathcal{L}(x,y)\in V\otimes \mathfrak{g}$ is the $L$-operator of the fundamental representation $V$. The Jimbo homomorphism can be decomposed as $\hat{\rho}_x= \rho\circ \pi_x$ where $\rho: \mathfrak{g}\rightarrow \text{End}(V)$. There is an important class of  homomorphisms $\eta$ which send $\mathfrak{g}$ to an infinite dimensional space $\mathcal{A}$. The homomorphisms $\rho$ and $\eta$ define the $R$-matrix and $L$-matrices\footnote{There are several homomorphisms $\mathfrak{g}\rightarrow \mathcal{A}$, however, for the purpose of the present paper we will be concerned with a specific one and call it $\eta$. The first example of such homomorphism for the algebra $\mathfrak{g}$ was given in \cite{Hayashi}.}
\begin{align}
&R(x,y) = \left( \text{id}\otimes \rho\right) \mathcal{L}(x,y), \label{Rmat}\\
&L(x,y) = \left(\text{id}\otimes \eta\right) \mathcal{L}(x,y). \label{Lmat}
\end{align}
If we  apply $\text{id}\otimes \text{id}\otimes \rho $ to (\ref{URLL}) we find the usual Yang--Baxter equation, while applying $\text{id}\otimes \text{id}\otimes \eta $ leads us to the RLL intertwining relation
\begin{align}
&R_{1,2}(x,y)R_{1,3}(x,z)R_{2,3}(y,z)=
R_{2,3}(y,z)R_{1,3}(x,z)R_{1,2}(x,y),\label{RYB}\\
&R_{1,2}(x,y)L_{1,3}(x,z)L_{2,3}(y,z)=
L_{2,3}(y,z)L_{1,3}(x,z)R_{1,2}(x,y).\label{RLL}
\end{align}
Recall the permutation operator $\mathcal{P}$ which transposes the tensor components $\mathcal{P} (a\otimes b) =(b\otimes a)$ for $a,b \in \hat{\mathfrak{g}}$, and is denoted by $P$ in the fundamental representation $V$. Acting with $P$ on (\ref{RYB}) and (\ref{RLL}) we get
\begin{align}
&\check{R}_{1,2}(x,y)R_{1,3}(x,z)R_{2,3}(y,z)=
R_{2,3}(y,z)R_{1,3}(x,z)\check{R}_{1,2}(x,y),\label{RcYB}\\
&\check{R}_{1,2}(x,y)L_{1,3}(x,z)L_{2,3}(y,z)=
L_{1,3}(y,z)L_{2,3}(x,z)\check{R}_{1,2}(x,y),\label{RcLL}
\end{align}
where $\check{R}(x,y)=P R(x,y)$.
We can omit indices in the last equation and recover (\ref{rll-eq})
\begin{align}
\check{R}(x,y)L(x,z)\otimes L(y,z)=
L(y,z)\otimes L(x,z)\check{R} (x,y).\label{RcLL2}
\end{align}

The universal $R$-matrix $\mathcal{R}$ of a quantum affine Lie algebra $A$, restricted to the trigonometric solutions \cite{Jimbo86b}, is defined by the commutation with the coproduct $\Delta: A\rightarrow A\otimes A$ of the algebra 
\begin{align}\label{RDel}
\mathcal{R}\Delta(a)=
\Delta'(a) \mathcal{R}, \qquad \forall a\in A,
\end{align}
where $\Delta'$ is the opposite coproduct $\mathcal{P}(\Delta(a))=\Delta'(a)$. Applying $\hat{\rho}_x\otimes \pi_y$ to (\ref{RDel}) we find the equation  
\begin{align}\label{LDel}
\mathcal{L}(x,y)\Delta(a)=
\Delta'(a) \mathcal{L}(x,y),
\end{align}
for the operator $\mathcal{L}$. For simplicity we keep the same symbol for the coproduct and its opposite. %%for $\hat{\rho}_x\otimes \pi_y \left(\Delta(a)\right)$ and $\hat{\rho}_x\otimes \pi_y \left(\Delta'(a)\right)$. 
The explicit form of the operator $\mathcal{L}$ for the Hopf algebra $\hat{\mathfrak{g}}$ with the standard coproduct $\Delta$ was given by Jimbo in \cite{Jimbo86a}. In the present paper we are dealing with stochastic vertex models which correspond to a twisted Hopf algebra $\hat{\mathfrak{g}}$ with the coproduct $\Delta^F$ twisted by an element $F$. More concretely, a twist $F$ is an invertible element (of a certain extension) of $\hat{\mathfrak{g}}\otimes \hat{\mathfrak{g}}$, written as 
\begin{align}\label{Ftwist}
F=\sum_i f_i \otimes f^i,
\end{align}
defines a new Hopf structure $\Delta^F$ of $\hat{\mathfrak{g}}$ (see \cite{Dr90,Resh90})
\begin{align}\label{DelF}
\Delta^F=F\Delta F^{-1}.
\end{align}
For an appropriately chosen $F$ the twisted Hopf algebra  $\hat{\mathfrak{g}}=U_{\tau}(A_r^{(1)})$ with the coproduct $\Delta^F$ defines the $U_{\tau}(A_r^{(1)})$  vertex models with a stochastic $R$-matrix. The resulting twisted Hopf algebra possesses the universal matrix $\mathcal{R}^F$ which, when restricted by $\hat{\rho}_x\otimes \pi_y$, gives the operator $\mathcal{L}^F(x,y)$ leading to the stochastic $R$-matrix and the associated $L$-matrices under the homomorphisms $\rho$ and $\eta$. Combining (\ref{RDel}) and (\ref{DelF}) we obtain
\begin{align*}
&F'\mathcal{R}F^{-1} \Delta^F(a) =
\Delta^F{}'(a)F' \mathcal{R}F^{-1},\\
&\mathcal{R}^F=F' \mathcal{R}F^{-1}.
\end{align*}
Set $\tilde{\Phi}=\hat{\rho}_x\otimes \pi_y (F')$ and $\Phi^{-1}=\hat{\rho}_x\otimes \pi_y (F^{-1})$, (\ref{LDel}) defines the $L$-operator
\begin{align}
&\mathcal{L}^F(x,y)\Delta^F(a)=
\Delta^F{}' (a)\mathcal{L}^F(x,y),\label{LFDel}\\
&\mathcal{L}^F(x,y)=\hat{\rho}_x\otimes \pi_y (\mathcal{R}^F)=\tilde{\Phi} \mathcal{L}(x,y) \Phi^{-1}, \label{FLF}
\end{align}
where we assumed again $\hat{\rho}_x\otimes \pi_y \left(\Delta^F(a)\right)$ and similarly for $\Delta^F{'}$.

In the next subsections we give the definitions of the algebras $\hat{\mathfrak{g}}$ and $\mathfrak{g}$, the Jimbo's formula for $\mathcal{L}$ and its twisting $\mathcal{L}^F$. We then give the explicit homomorphism $\eta$ and write the resulting matrix $L(x,y)$ together with the stochastic matrix $R(x,y)$ as images of  $\eta$ and $\rho$ applied to $\mathcal{L}^F$, respectively. 

\subsection{Definitions}\label{subsect11}
Let us recall the construction of Jimbo \cite{Jimbo86a}. Fix a parameter $\tau\in \mathbb{C}$, $|\tau|<1$. We start with the algebra $U_{\tau}(sl_{r+1})$ generated by the elements $e_i, f_i$ and $K_i,K^{-1}_i$ ($i=1,\dots,r$) which satisfy the following relations 
\begin{align}\label{CRsl}
&K_i K^{-1}_i =K^{-1}_i  K_i =1, \qquad [K_i,K_j]=0,  \nonumber\\ 
&K_i e_j  = \tau^{ C_{i,j}}e_j K_i, \qquad  K_i f_j  = t^{-C_{i,j}}f_j K_i,  \nonumber \\
&[e_i,f_j]= \delta_{i,j} \frac{K_i-K_i^{-1}}{\tau-\tau^{-1}},  \nonumber\\
&e_i^2 e_{i \pm 1}- (\tau+{\tau^{-1}})  e_{i}e_{i \pm 1} e_{i}+ e_{i\pm 1} e_i^2 =0,\quad 1\leq i,i\pm 1 \leq r \nonumber \\
&f_i^2 f_{i \pm 1}- (\tau+{\tau^{-1}})  f_{i}f_{i \pm 1} f_{i}+ f_{i\pm 1} f_i^2 =0,\quad 1\leq i,i\pm 1 \leq r. 
\end{align}
Where $C_{i,j}$ are the matrix elements of the Cartan matrix $C_{i,i}=2$, $C_{i,i+1}=C_{i,i-1}=-1$ and $C_{i,j}=0$ for $|i-j|>1$. We define the algebra $\mathfrak{g}=U'_{\tau}(sl_{r+1})$ by adding the elements 
$\kappa_i^{\pm 1}$ ($i=0,\dots,r$) to the algebra $U_{\tau}(sl_{r+1})$ such that $K_i=\kappa_{i-1}\kappa_{i}^{-1}$. It is convenient to introduce also the elements $\epsilon_i$ ($i=0,\dots,r$), related to $\kappa_i$ by $\kappa_i=\tau^{\epsilon_i}$. The new elements $\kappa_j$ commute with $e_i$ and $f_i$ for $j<i-1$ and $j>i$, otherwise
\begin{align*}
&\kappa_i e_i =\tau^{-1}e_i \kappa_i,	&\quad& 		\kappa_i f_i =\tau f_i \kappa_i, \\
&\kappa_{i-1} e_i =\tau e_i \kappa_{i-1},&\quad&		\kappa_{i-1} f_i =\tau^{-1} f_i \kappa_{i-1}.
\end{align*}
The element $\kappa_0 \kappa_1 \cdots \kappa_r =c$ belongs to the center of the algebra. The algebra $\mathfrak{g}$ is equipped with the coproduct $\Delta : U \rightarrow \mathfrak{g}\otimes \mathfrak{g}$   
\begin{align*}
&\Delta (e_i) = K_i^{1/2}  \otimes e_i+e_i  \otimes K_i^{-1/2},\qquad \Delta (f_i) = K_i^{1/2}  \otimes f_i+f_i  \otimes K_i^{-1/2}, \qquad \Delta (\kappa_i^{\pm 1}) = \kappa_i^{\pm 1} \otimes \kappa_i^{\pm 1}.
\end{align*}
Next we define the higher root elements $e_{i,j}\in \mathfrak{g}$ ($0\leq i\neq j \leq r$)
\begin{align}
&e_{i-1,i} = e_i,\qquad e_{i,i-1} = f_i, \label{hroots1} \\
&e_{i,j} = e_{i,k}e_{k,j} - \tau  e_{k,j}e_{i,k}, &\quad& \text{for}~i>j, \label{hroots2}\\
&e_{i,j} = e_{i,k}e_{k,j} - \tau^{-1}  e_{k,j}e_{i,k}, &\quad& \text{for}~i<j.\label{hroots3}
\end{align}
It is easy to verify that the higher roots $e_{i,j}$ commute with $\kappa_l$ for $l\neq i,j$ and otherwise 
\begin{align*}
&\kappa_i e_{i,j} =\tau e_{i,j} \kappa_i , &\quad& 	\kappa_j e_{i,j} =\tau^{-1} e_{i,j} \kappa_j,  	&\quad& {i<j},\\
&\kappa_i e_{i,j} =\tau^{-1} e_{i,j} \kappa_i , &\quad& 	\kappa_j e_{i,j} =\tau e_{i,j} \kappa_j,   &\quad& {i>j}.
\end{align*}
They also enjoy the Cartan--Weyl basis type property
\begin{align*}
&[e_{i,j},e_{j,i}]= \frac{\kappa_{i}\kappa_j^{-1}-\kappa_{i}^{-1}\kappa_j}{\tau-\tau^{-1}}.  \nonumber\\
\end{align*}
The quantum affine Lie algebra $\hat{\mathfrak{g}}$ is generated by the elements $\tilde{e}_i, \tilde{f}_i$ and $\tilde{K}_i^{\pm 1}$ ($i=0,\dots,r$) satisfying (\ref{CRsl}) with the Cartan matrix $\hat{C}_{i,j}$ of the algebra $A_r^{(1)}$. The algebra homomorphism $\pi_x: \hat{\mathfrak{g}} \rightarrow  \mathfrak{g} \otimes \mathbb{C}[x,x^{-1}]$ is given by
\begin{align}
&\pi_x(\tilde{e}_0)=x e_0,	
&\quad& \pi_x(\tilde{f}_0)= x^{-1}f_0,	
&\quad& \pi_x(\tilde{K}_0)= \kappa_r\kappa_0^{-1},\\
&\pi_x(\tilde{e}_i)=e_i,	&\quad& \pi_x(\tilde{f}_i)=f_i, &\quad& \pi_x(\tilde{K}_i)=K_i  \quad \text{for}~i>0,\\
&e_0=\tau \kappa_0^{-1}\kappa_r^{-1}e_{r,0},&\qquad& f_0=\tau^{-1} \kappa_0 \kappa_r e_{0,r}.
\end{align}

\subsection{$L$-operator}\label{subsect12}
The fundamental representation $\rho: \mathfrak{g}\rightarrow \text{End}(V)$ with $V=\mathbb{C}^{r+1}$ is defined by  
\begin{align}
e_{i,j}\mapsto E_{i,j}, \qquad \kappa_i \mapsto \mathbb{I}+(\tau-1)E_{i,i},  \qquad \epsilon_i \mapsto E_{i,i},
\end{align}
where $ \mathbb{I}$ is the identity matrix in $V$, $E_{i,j}$ are the matrix units with one at position $i,j$ and zero elsewhere and the indices $i,j$ run from $0$ to $r$. Assuming that $\Delta(a)$ is the image of the coproduct of $\mathfrak{g}$ under $\rho\otimes \text{Id}$, then the $L$-operator $\mathcal{L}$ of the fundamental representation satisfies the following linear equations 
\begin{align}
&\mathcal{L}(x,y)\Delta(a)=\Delta'(a) \mathcal{L}(x,y), \quad a\in \mathfrak{g},\label{LDel1}\\
&\mathcal{L}(x,y)D(e_0)=D'(e_0)\mathcal{L}(x,y), \label{LDel2}\\
&\mathcal{L}(x,y)D(f_0)=D'(f_0)\mathcal{L}(x,y),\label{LDel3}
\end{align}
where we defined 
\begin{align*}
&D(e_0)=y \kappa_r^{1/2}\kappa_0^{-1/2} \otimes e_0 + x e_0\otimes \kappa_r^{-1/2}\kappa_0^{1/2},\qquad
D'(e_0)= y \kappa_r^{-1/2}\kappa_0^{1/2} \otimes e_0 + x e_0\otimes \kappa_r^{1/2}\kappa_0^{-1/2}, \\
&D(f_0)= x \kappa_r^{1/2}\kappa_0^{-1/2} \otimes f_0 + y f_0\otimes \kappa_r^{-1/2}\kappa_0^{1/2},\qquad 
D'(f_0)= x \kappa_r^{-1/2}\kappa_0^{1/2} \otimes f_0 + y f_0\otimes \kappa_r^{1/2}\kappa_0^{-1/2}.
\end{align*}
In \cite{Jimbo86a} it was shown that (\ref{LDel1})--(\ref{LDel3}) can be solved by the $L$-operator $\mathcal{L}(x,y)$, which we write in the from adopted for our purposes
\begin{align*}
&\mathcal{L}(x,y)=\sum_{0\leq i,j\leq r} E_{j,i}\otimes \hat{E}_{i,j}(x,y), \\
&\hat{E}_{i,j}(x,y)=
\left\{
\begin{array}{ll}
\tau^{-1/2}~x~\kappa_i^{1/2}\kappa_j^{1/2} e_{i,j},  \qquad  &  i<j
\\
\\
(\tau-\tau^{-1})^{-1} \left( x \kappa_i- y \kappa_i^{-1}\right),  \qquad  & i=j
\\
\\
\tau^{1/2}~y ~\kappa_i^{-1/2}\kappa_j^{-1/2} e_{i,j},  \qquad  &   i>j.
\end{array}
\right.
\end{align*}

As discussed above, the next step towards the stochastic $R$-matrix and the associated $L$-matrix is to deform the Hopf structure of the algebra with a twist $F$. We choose the twist to be
\begin{align}
F=\tau^{-\sum_{j<i} \epsilon_i\otimes \epsilon_j}.
\end{align}
Applying the homomorphism $\hat{\rho}_x\otimes \pi_y$ we find
\begin{align}
&\tilde{\Phi}=\hat{\rho}_x\otimes \pi_y (F')=\sum_{j=0}^r E_{j,j} \otimes \prod_{l=j+1}^r \kappa_l^{-1} \\
&\Phi=\hat{\rho}_x\otimes \pi_y (F^{-1})=\sum_{i=0}^r E_{i,i} \otimes\prod_{l=0}^{i-1} \kappa_l^{-1}.
\end{align}
The operator $\mathcal{L}^F(x,y)$ is a solution of (\ref{LDel1})-(\ref{LDel3}) were $\Delta$, $\Delta'$, $D$ and $D'$ are substituted with their twisted versions, implying the homomorphism $\hat{\rho}_x\otimes \pi_y$ we have 
\begin{align}
&\Delta^F= \Phi \Delta \Phi^{-1},	&\qquad&  \Delta^F{}'= \tilde{\Phi} \Delta' \tilde{\Phi}^{-1}, \\
&D^F(X)= \Phi D(X) \Phi^{-1},	&\qquad&  D^F{}'(X)= \tilde{\Phi} D'(X) \tilde{\Phi}^{-1}, \qquad X=e_0,f_0.
\end{align}
With these redefinitions equations (\ref{LDel1})-(\ref{LDel3}) lead to the following form for the operator $\mathcal{L}^F(x,y)$
\begin{align}
&\mathcal{L}^F(x,y)=\tilde{\Phi} \mathcal{L}(x,y) \Phi^{-1}=\sum_{0\leq i,j\leq r} E_{j,i}\otimes  \hat{E}^F_{i,j}(x,y), \label{Lstoch} \\
&\hat{E}^F_{i,j}(x,y)= \prod_{l=j+1}^r \kappa_l^{-1} \hat{E}_{i,j}(x,y) \prod_{l=0}^{i-1} \kappa_l. \nonumber
\end{align}
Using the fact that $\kappa_0 \kappa_1 \cdots \kappa_r =c$ we can write $\hat{E}^F_{i,j}(x,y)$ as 
\begin{align}\label{Estoch}
\hat{E}^F_{i,j}(x,y)=c\times 
\left\{
\begin{array}{ll}
 \tau^{-1/2}~x~\kappa_i^{1/2}\kappa_j^{1/2} \prod_{l=j+1}^{r}\kappa_l^{-2} \prod_{l=i}^j\kappa_l^{-1} e_{i,j},  \qquad  &  i<j
\\
\\
(\tau-\tau^{-1})^{-1} \left( x- y  \kappa_i^{-2}\right)\prod_{l=i+1}^{r}\kappa_l^{-2},  \qquad  & i=j
\\
\\
\tau^{3/2}~y ~\kappa_i^{-1/2}\kappa_j^{-1/2} \prod_{l=i}^r\kappa_l^{-2} \prod_{l=j+1}^{i-1}\kappa_l^{-1} e_{i,j},  \qquad  &   i>j.
\end{array}
\right.
\end{align}

\subsection{Homomorphism to the $t$-Oscillator algebra}\label{subsect13}
In this section we provide a homomorphism $\eta$ which takes $\mathfrak{g}$ to $\mathcal{A}=\mathcal{A}_1\otimes\dots\otimes \mathcal{A}_r$, where $\mathcal{A}_i$ is the algebra  of $\tau^2$-oscillators $\mathcal{A}_i=\{a_i,a_i^{\dag},h_i\}$ with defining relations 
\begin{align}
&h_i a_i = \tau^{-2}a_i h_i,\qquad 	h_i a^{\dag}_i = \tau^2 a^{\dag}_i h_i, \nonumber \\
&a_i a_i^{\dag} = 1-\tau^2 h_i,\qquad a_i^{\dag} a_i= 1-h_i.
\label{tosc}
\end{align}
The homomorphism $\eta$ reads 
\begin{align}
& \eta(\kappa_0 )= c^{1/2} \prod_{i=1}^r h_i^{1/2},									&\qquad& 
\eta(\kappa_i )=h_i^{-1/2}, \qquad \text{for}~i=1,\dots,r, \\
&\eta (e_1) = \frac{c^{-1/2}\tau^{1/2}}{\tau-\tau^{-1}} h_0^{-1/4}h_1^{-1/4} a^{\dag}_1, &\qquad&
\eta (f_1) = \frac{c^{1/2}\tau^{-3/2}}{\tau-\tau^{-1}} h_0^{-1/4}h_1^{-1/4}(1-h_0) a_1, \\
& \eta (e_i) = \frac{\tau^{1/2}}{\tau-\tau^{-1}} h_{i-1}^{-1/4}h_{i}^{-1/4} a_{i-1} a^{\dag}_i, &\qquad&
 \eta (f_i) = \frac{\tau^{-3/2}}{\tau-\tau^{-1}} h_{i-1}^{-1/4}h_{i}^{-1/4} a^{\dag}_{i-1}a_{i}, \quad \text{for}~i=2,\dots,r,
\label{oschom}
\end{align}
where we used the following convenient notation
\begin{align}
h_0=c^{-1} \prod_{i=1}^r h_i^{-1}, \qquad a_0^{\dag}=c^{1/2}(1-h_0),\qquad a_0=c^{-1/2}. \label{Azero}
\end{align}
Note, $\{a_0,a_0^{\dag},h_0\}$ do not satisfy the relations of the $\tau$-oscillator algebra. With this notation we can write compactly the homomorphism $\eta$ for the simple roots $e_1=e_{0,1}$, $f_1=e_{1,0}$, $e_i=e_{i-1,i}$ ($i>1$) and $f_i=e_{i,i-1}$ ($i>1$) and for the higher roots $e_{i,j}$
\begin{align}
&\eta (e_{i,j}) = \frac{\tau^{1/2}}{\tau-\tau^{-1}}h_i^{-1/4}h_j^{-1/4}\prod_{i< l<j} h_l^{-1/2} a_i a^{\dag}_j,  &\quad& \text{for}~i<j\\
&\eta (e_{i,j}) = \frac{\tau^{-3/2}}{\tau-\tau^{-1}}h_i^{-1/4}h_j^{-1/4}\prod_{j< l<i} h_l^{1/2} a^{\dag}_j a_i,  &\quad& \text{for}~i>j,
\end{align}
where $i,j$ run from $0$ to $r$.
It is easy to verify the validity of these equations using (\ref{hroots1})--(\ref{hroots3}). Previously we were using the $t$-oscillator algebras $\mathcal{B}_i$ generated by $\{\phi_i,\phi_i^{\dag},k_i\}$, satisfying 
\begin{align*}
&k_i \phi_i = t^{-1}\phi_i k_i,\qquad 	k_i \phid_i = t \phid_i k_i, \nonumber \\
&\phi_i \phi_i^{\dag} = 1-t k_i,\qquad \phi_i^{\dag} \phi_i= 1-k_i.
\end{align*}
The algebras $\mathcal{B}_i$ are related with $\mathcal{A}_i$ simply by identifying $\phi_i=a_i$, $\phi^{\dag}_i=a^{\dag}_i$, $k_i=h_i$ and $t=\tau^2$. One also needs to add the notation $\mathcal{B}_0$
\begin{align}
k_0=c^{-1} \prod_{i=1}^r k_i^{-1}, \qquad \phi_0^{\dag}=c^{1/2}(1-k_0),\qquad \phi_0=c^{-1/2}. \label{Bzero}
\end{align}
Applying the fundamental homomorphism $\rho$ and oscillator homomorphism $\eta$ in (\ref{Lstoch}) and (\ref{Estoch}) we obtain the $R$-matrix and the $L$-matrix
\begin{align}
&R(x,y)=\text{id}\otimes \rho\left(\mathcal{L}^F(x,y)\right)=c~ t^{-1/2}\times \sum_{0\leq i,j\leq r} E_{j,i}\otimes  R_{i,j}(x,y), \nonumber \\
&L(x,y)=\text{id}\otimes \eta\left(\mathcal{L}^F(x,y)\right)=\frac{c}{(t^{1/2}-t^{-1/2})}\times \sum_{0\leq i,j\leq r} E_{j,i}\otimes  L_{i,j}(x,y), \nonumber
\end{align}
\begin{align}
&R_{i,j}(x,y)= 
\left\{
\begin{array}{ll}
x ~E_{i,j},  \qquad  &  i<j
\\
\\
\frac{t(x-y)}{t-1}\mathbb{I}+y\sum_{l=i}^rE_{l,l}-x\sum_{l=i+1}^rE_{l,l},  \qquad  & i=j
\\
\\
y~ E_{i,j} ,  \qquad  &   i>j.
\end{array}
\right.\\
&L_{i,j}(x,y)= 
\left\{
\begin{array}{ll}
x~\prod_{l=j+1}^{r}k_l ~\phi_i \phi^{\dag}_j,  \qquad  &  i<j
\\
\\
 \left( x- y  k_i\right)\prod_{l=i+1}^{r}k_l,  \qquad  & i=j
\\
\\
y ~ \prod_{l=j+1}^r k_l ~\phi^{\dag}_j \phi_i ,  \qquad  &   i>j.
\end{array}
\right. 
\end{align}
The above matrix $R(x,y)$ is the same as (\ref{Rmat-def}) up to the overall normalisation $c~t^{-1/2}$. In order to match $L(x,y)$ with (\ref{eq:Lmat}) we must recall the definition of $\mathcal{B}_0$ (\ref{Bzero}), set $c= u v$ and $y=v$, multiply the first column by $(u v)^{1/2}$ and the first row by $-(u/v)^{1/2}$ and finally normalise the matrix by the factor $(t^{1/2}-t^{-1/2})^{-1}(u v)$.

\section{Conclusion}
This paper contains three major results. First, we describe a general scheme for the construction of symmetric polynomials using a matrix product formalism. The main idea of this construction is to first define a family of polynomials $f_\mu$, indexed by compositions $\mu$, which solve the local quantum Knizhnik--Zamolodchikov (qKZ) exchange relations. Such polynomials $f_\mu$ are non-symmetric and can be expressed as matrix products if a certain solution to the Zamolodchikov--Faddeev (ZF) algebra can be found. A symmetric polynomial is then obtained by symmetrisation over the family $f_\mu$, similar to what occurs in the theory of non-symmetric Macdonald polynomials $E_\mu$. We emphasise here that the polynomials $f_\mu$ constructed in this paper are quite different from $E_\mu$.

Our second result is a general and constructive method to obtain solutions to the ZF algebra from $L$-matrix solutions to the Yang-Baxter equation, and the third result is a new bosonic $L$-matrix. The latter is obtained using a new homomorphism from the quantum group to families of deformed oscillators.

Using these three results, the approach outlined above culminates in a new family of polynomials that generalise Macdonald polynomials, and unifies these with another class of polynomials recently studied by Borodin and Petrov. By virtue of their construction we expect our new family to have a number of natural properties such as Cauchy identities, branching rules and Pieri identities.

\section*{Acknowledgments}

We gratefully acknowledge support from the Australian Research Council Centre of Excellence for Mathematical and Statistical Frontiers (ACEMS), and MW acknowledges support by an Australian Research Council DECRA. JdG would like to thank the KITP Program \textit{New approaches to non-equilibrium and random systems: KPZ integrability, universality, applications and experiments}, supported in part by the National Science Foundation under Grant No. NSF PHY11-25915, and with MW the program \textit{Statistical mechanics and combinatorics} at the Simons Center for Geometry and Physics, Stony Brook University, where part of this work were completed. MW would like to thank Alexei Borodin for kind hospitality at MIT and very stimulating discussions on related themes.


\begin{thebibliography}{99}

\bibitem{AritaAMP} C. Arita, A. Ayyer, K. Mallick, S. Prolhac, \textit{Generalized matrix Ansatz in the multispecies exclusion process - partially asymmetric case}, J. Phys. A: Math. Theor. \textbf{45} (2012), 195001.

\bibitem{Borodin-com} A. Borodin, Private communication.

\bibitem{Borodin} A. Borodin, \textit{On a family of symmetric rational functions}, arXiv:1410.0976.

\bibitem{BorodinC} A. Borodin and I. Corwin, \textit{Macdonald processes}, Probability Theory and Related Fields {\bf 158} (2014), 225--400.

\bibitem{BorodinP} A. Borodin and L. Petrov, \textit{Higher spin six vertex model and symmetric rational functions},  arXiv:1601.05770.

\bibitem{CantinidGW} L.~Cantini, J.~de~Gier and M.~Wheeler, \textit{Matrix product formula for Macdonald polynomials}, J. Phys. A: Math. Theor. {\bf 48} (2015), 384001; arXiv:1505.00287.

\bibitem{Cher95a} I. Cherednik, \textit{Double affine Hecke algebras and Macdonald's conjectures}, Annals Math. \textbf{141} (1995), 191--216.

\bibitem{Cher95b} I. Cherednik, \textit{Nonsymmetric Macdonald polynomials}, Internat. Math. Res. Notices \textbf{10} (1995), 483--515.

\bibitem{CorwinP} I. Corwin and L. Petrov, \textit{Stochastic higher spin vertex models on the line}, arXiv:1502.07374.

\bibitem{CrampeRV} N. Crampe, E. Ragoucy, M. Vanicat, \textit{Integrable approach to simple exclusion processes with boundaries. Review and progress}, J. Stat. Mech. (2014), P11032;  arXiv:1408.5357.

\bibitem{Dr87} V. G. Drinfeld, \textit{Quantum groups}, Proc. ICM-86 (Berkeley, USA) vol. 1, Amer. Math. Soc., (1987), 798--820.

\bibitem{Dr90} V. G. Drinfeld, \textit{Quasi-Hopf algebra. Algebra and Analysis}, Peterburg Math. Journ., (1990), 1419--1457.

\bibitem{Fad1980} L. D. Faddeev, \textit{Quantum completely integrable models in field theory}, in: Problems of Quantum Field Theory, R2-12462, Dubna (1979), 249--299. L. D. Faddeev, \textit{Quantum completely integrable models in field theory}, in: Contemporary Mathematical Physics, Vol. IC (1980), 107--155.

\bibitem{FodaW} O. Foda and M. Wheeler, \textit{Colour-independent partition functions in coloured vertex models}, Nucl. Phys. B {\bf 871} (2013), 330--361.

\bibitem{Haiman} M. Haiman, \textit{Hilbert schemes, polygraphs and the Macdonald positivity conjecture}, J. Amer. Math. Soc. \textbf{14} (2001), 941--1006.

\bibitem{Hayashi} T. Hayashi, \textit{Q-analogues of Clifford and Weyl algebras-spinor and oscillator representations of quantum enveloping algebras}, Comm. Math. Phys. {\bf127} (1990), 129--144.

\bibitem{InoueKO} R. Inoue, A. Kuniba and M. Okado, \textit{A quantization of box-ball systems}, Rev. Math. Phys. \textbf{16} (2004), 1227--1258; arXiv:nlin/0404047.

\bibitem{Jimbo86a} M. Jimbo, \textit{A $q$-analogue of $U(gl(N+1))$, Hecke algebra, and the Yang--Baxter equation}, Lett. Math. Phys. {\bf 11} (1986), 247--252.

\bibitem{Jimbo86b} M. Jimbo, \textit{Quantum $R$-matrix for the generalized Toda system}, Comm. Math. Phys. {\bf 102} (1986), 537--547.

\bibitem{KasataniT} M. Kasatani and Y. Takeyama, \textit{The quantum Knizhnik--Zamolodchikov equation and non-symmetric Macdonald polynomials}, Funkcialaj ekvacioj. Ser. Internacia \textbf{50} (2007), 491--509; arXiv:math/0608773.

\bibitem{Korff} C. Korff, \textit{Cylindric versions of specialised Macdonald polynomials and a deformed Verlinde algebra}, Comm. Math. Phys. {\bf 318} (2013), 173--246.

\bibitem{KunibaMMO} A. Kuniba, V. V. Mangazeev, S. Maruyama and M. Okado, \textit{Stochastic $R$ matrix for $U_q(A^{(1)}_n)$}, arXiv:1604.08304.
 
\bibitem{Macd88} I. Macdonald, \textit{A new class of symmetric functions}, Publ. I.R.M.A. Strasbourg, Actes 20$^\textrm{e}$ S\'eminaire Lotharingien \textbf{131-71} (1988).

\bibitem{MacdBook} I. Macdonald, \textit{Symmetric functions and Hall polynomials}, (2nd ed.), Oxford, Clarendon Press 1995.

\bibitem{Mang14a} V. Mangazeev, \textit{On the Yang-Baxter equation for the six-vertex model}, Nucl. Phys. B \textbf{882} (2014) 70--96.

\bibitem{Mang14b} V. Mangazeev, \textit{Q-operators in the six-vertex model}, Nucl. Phys. B \textbf{886} (2014) 166--184.

\bibitem{Opdam} E. Opdam, \textit{Harmonic analysis for certain representations of graded Hecke algebras}, Acta Math. \textbf{175} (1995), 75--121.

\bibitem{Povo13} A. M. Povolotsky, \textit{On the integrability of zero-range chipping models with factorized steady states}, J. Phys. A: Math. Theor. \textbf{46} (2013) 465205 (25pp).

 \bibitem{ProlhacEM} S. Prolhac, M. R. Evans, K. Mallick, \textit{Matrix product solution of the multispecies partially asymmetric exclusion process}, J. Phys. A: Math. Theor. \textbf{42} (2009), 165004; arXiv:0812.3293.

\bibitem{Resh90} N. Yu. Reshetikhin, \textit{Multiparameter quantum groups and twisted quasitriangular Hopf algebras}, Lett. Math. Phys \textbf{20} (1990), 331--335

\bibitem{SasaW98} T. Sasamoto and M. Wadati, \textit{Exact results for one-dimensional totally asymmetric diffusion models}, J. Phys. A: Math. Gen. \textbf{31} (1998) 6057--6071.

\bibitem{Takey14} Y. Takeyama, \textit{A deformation of affine Hecke algebra and integrable stochastic particle system}, J. Phys. A: Math. Theor. \textbf{47} (2014), 465203 (19pp).

\bibitem{Takey15} Y. Takeyama, \textit{Algebraic construction of multi-species q-Boson system}, arXiv:1507.02033.

\bibitem{Tsuboi} Z. Tsuboi, \textit{Asymptotic representations and $q$-oscillator solutions of the graded Yang--Baxter equation related to Baxter $Q$-operators}, Nucl. Phys. B 886 (2014), 1--30; arXiv:1205.1471.

\bibitem{WheelerZJ} M. Wheeler and P. Zinn-Justin, \textit{Refined Cauchy/Littlewood identities and six-vertex model partition functions: III. Deformed bosons}, To appear in Adv. Math., arXiv:1508.02236.

\bibitem{ZZ1979} A. B. Zamolodchikov and AI. B. Zamolodchikov, \textit{Two-dimensional factorizable S-matrices as exact solutions of some quantum field theory models}, Ann. Phys. \textbf{120} (1979), 253--291.


\end{thebibliography}
\end{document}